\documentclass[pdflatex,sn-mathphys-num]{sn-jnl}
\usepackage{graphicx}%
\usepackage{multirow}%
\usepackage{amsmath,amssymb,amsfonts}%
\usepackage{amsthm}%
\usepackage{mathrsfs}%
\usepackage[title]{appendix}%
\usepackage{xcolor}%
\usepackage{textcomp}%
\usepackage{manyfoot}%
\usepackage{booktabs}%
\usepackage{algorithm}%
\usepackage{algorithmicx}%
\usepackage{algpseudocode}%
\usepackage{listings}%
\usepackage{float}%

\usepackage[utf8]{inputenc}
\usepackage{url}
\usepackage{hyperref}
\usepackage{color}
\usepackage[T1,T2A]{fontenc}

\date{June 2026}

\usepackage{mathrsfs}
\usepackage{csquotes}
\usepackage{graphicx}
\usepackage{dcolumn}	
\usepackage{bm}			
\usepackage{amsfonts}
\usepackage{xspace}
\usepackage{epstopdf}
\usepackage{multirow}
\usepackage{braket}
\usepackage{amssymb}
\usepackage{amsthm}
\usepackage{textcomp}
\usepackage{mathtools}
\usepackage{latexsym}
\usepackage{amsmath}

\newtheorem{theorem}{Theorem}[section]

\newtheorem{prop}[theorem]{Proposition}
\newtheorem{cor}[theorem]{Corollary}

\makeatletter
\newcommand{\vast}{\bBigg@{2}}
\newcommand{\Vast}{\bBigg@{3}}

\makeatother

\newtheorem{conj}{Conjecture}

\theoremstyle{definition}
\newtheorem{defn}[theorem]{Definition}
\newtheorem{ex}[theorem]{Example}

\begin{document}

\title[Article Title]{The relationship between spacetime singularities and regions at infinity}


\author{Junbang Liu}\email{Junbang.Liu@anu.edu.au}

\author{Ben Andrews}\email{Ben.Andrews@anu.edu.au}

\author{Susan M.\ Scott}\email{Susan.Scott@anu.edu.au}

\affil{The Australian National University, \orgaddress{ Canberra ACT 2601, Australia}}


\abstract{Ideal attached points are a core concept in general relativity for pseudo-Riemannian manifolds, and whether the spacetime can be extended with certain properties is a central consideration in their choice. This paper establishes a sufficient condition for the separability between singularities and points at infinity for any maximally extended pseudo-Riemannian manifold. We focus on the incomplete geodesics of $(\mathcal{M},g)$, and produce an envelopment $(\mathcal{M},g,\hat{\mathcal{M}})$ of the spacetime such that an incomplete geodesic $\gamma:[0,1) \rightarrow \mathcal{M}$ has an endpoint $q$ in $\hat{\mathcal{M}}$. If there is no pair of geodesics approaching $q$ which is intertwined, then $q$ is a singularity. Additionally, $q$ will not be approached by any geodesic with infinite affine parameter, and therefore cannot cover a point at infinity, thereby rendering it a {\it pure singularity} in the abstract boundary framework. We apply the Endpoint Theorem to the maximal g-boundary introduced by Graf and Beld-Serrano in \cite{Graf_2024}, and also provide a result on the separability between directional singularities and pure singularities. This analysis is then applied to the Schwarzschild spacetime.}

\keywords{abstract boundary construction, pseudo-Riemannian geometry, envelopment, maximal extension, singularity, region at infinity, separability}



\maketitle

\section{Introduction}
In general relativity one often seeks to extend spacetimes under consideration. The extension of a spacetime is intimately connected to the question of how to attach a boundary to a spacetime. The topological structure of singularities related to geodesic incompleteness has been widely explored via Geroch's g-boundary \cite{geroch1968local}, Schmidt's b-boundary \cite{schmidt1971new}, Geroch, Kronheimer and Penrose's c-boundary \cite{geroch1972ideal,garcia2005causal}, and Scott and Szekeres's abstract boundary (a-boundary) \cite{scott1994abstract}. The a-boundary is a flexible scheme to classify boundary points under different embeddings of spacetime.\\

For many interesting results related to globally hyperbolic spacetimes and the vacuum Einstein field equation, asymptotically flat initial data is assumed as a physical condition. It is an interesting question as to whether the regions at infinity and the singularities should be disjoint in this framework. It remains an important open question, in general, as to \textit{how to separate singularities and points at infinity in embeddings?} This question has a well-framed description in the abstract boundary context, namely the issue of the separability between pure singularities and pure points at infinity (i.e.\ whether there exist open neighbourhoods of the boundary points of $\mathcal{M}$ in embeddings $\psi_1: \mathcal{M} \rightarrow \hat{\mathcal{M}}$ and $\psi_2: \mathcal{M} \rightarrow \mathcal{M}'$ such that these open neighbourhoods of the pure singularity and the pure point at infinity respectively, restricted to $\mathcal{M}$, are disjoint).\\

Using the abstract boundary construction, in this paper we will only consider embeddings $\psi:\mathcal{M} \rightarrow \hat{\mathcal{M}}$ of $n$-dimensional, pseudo-Riemannian manifolds $(\mathcal{M},g)$ which are maximally extended. In relation to the singularity theorems, this involves no loss of generality, as they also assume that the spacetime is maximally extended \cite{whale2015generalizations}.\\


For the abstract boundary construction, if there exists an embedding $\psi: \mathcal{M} \rightarrow \hat{\mathcal{M}}$ and $p \in \partial\psi(\mathcal{M})$ such that $p$ is only approached by curves with an unbounded affine parameter, we say that $p$ is a \textit{point at infinity}. For a singularity, if there exists an embedding $\psi: \mathcal{M} \rightarrow \hat{\mathcal{M}}$ and $p \in \partial\psi(\mathcal{M})$ such that $p$ is approached by some curves with bounded affine parameter, then $p$ is an \textit{essential singularity} \cite{scott1994abstract}. A \textit{pure singularity} is an essential singularity that cannot be re-embedded into a non-singular boundary point in any other embedding. The relevant details of the abstract boundary will be introduced formally in Section \ref{section 2}. The following theorem provides the motivation for us to directly investigate maximally extended pseudo-Riemannian manifolds.
\begin{theorem}\label{theorem:1}
For a maximally extended pseudo-Riemannian manifold $(\mathcal{M},g)$, boundary points of envelopments of $(\mathcal{M},g)$  have the following properties: \begin{enumerate}
        \item A point at infinity must be a pure point at infinity
        \item A singularity can only be either a directional singularity or a pure singularity
        \item A directional singularity must cover a pure point at infinity
        \item A pure point at infinity/pure singularity can only cover pure points at infinity/pure singularities, respectively, and non-regular unapproachable boundary points.
    \end{enumerate}
\end{theorem}
\begin{proof}
    The core assumption is that $(\mathcal{M},g)$ is maximally extended so that there does not exist any extension $(\mathcal{M}',g')$ of $(\mathcal{M},g)$, and thus there are no regular boundary points.
\end{proof}


For any maximally extended pseudo-Riemannian manifold, Theorem \ref{theorem:1} simplifies the separability problem between directional singularities and pure singularities to the separability problem between pure points at infinity and pure singularities. The question then reduces to the asymptotic behaviour of some specific curves with finite length and some with infinite length.
The Endpoint Theorem \cite{Scott_2021} is essential for the solution of this question.
\begin{theorem}(The Endpoint Theorem)\label{theorem:endpoint}
    Let $\mathcal{M}$ be an $n$-dimensional, smooth, connected, Hausdorff, paracompact manifold. If $(x_i)_{i \in \mathbb{N}}$ is a sequence of points in $\mathcal{M}$ without an accumulation point, then there always exists an $n$-dimensional, smooth, connected, Hausdorff, paracompact manifold $\mathcal{N}$ and an open embedding $\psi : \mathcal{M} \rightarrow \mathcal{N}$, such that $\partial\psi(\mathcal{M})$ is diffeomorphic to the $n-1$ dimensional unit ball and the sequence $(\psi(x_i))_{i \in \mathbb{N}}$ converges to some $y \in \partial\psi(\mathcal{M})$.
\end{theorem}

The exhaustion method used in the proof of the Endpoint Theorem \cite{Scott_2021} works for any topological manifold (that is, any manifold with $C^l$ charts for all $l \geq 0$). The normal neighbourhood construction employed requires the existence of geodesics, so the Endpoint Theorem is applicable to any $C^k$ manifold for all $k \geq 2$.\\

For any sequence of points without an accumulation point in an $n$-dimensional manifold $\mathcal{M}$, the Endpoint Theorem constructs an embedding of $\mathcal{M}$ into a larger $n$-dimensional manifold $\mathcal{N}$ such that the sequence has an accumulation point, indeed an endpoint, in the boundary set of this embedding. It does not only take into account how the geometry is extended through the boundary, but also how the boundary is globally attached to the manifold.

\subsection{Main results and application to the maximal g-boundary}
For the normal neighbourhood construction in the proof of Theorem \ref{theorem:endpoint}, one has the flexibility to adjust the neighbourhood via the size function $f: [0,1) \rightarrow \mathbb{R}^+$ (a smooth function). A sufficiently small normal neighbourhood $U$ of a geodesic can be chosen such that the geodesic has an endpoint in the new embedding which is not the endpoint of any other geodesic. This technique allows us to restrict geodesics approaching the targeted boundary point in $\mathcal{N}$. With the abstract boundary classification for any maximally extended pseudo-Riemannian manifold $(\mathcal{M},g)$, it implies that the endpoint in $\mathcal{N}$ of the central geodesic with bounded affine parameter is very likely to be a pure singularity.\\

Limiting behaviour, however, also plays an essential role in this classification. There exist some types of curves which cannot be isolated by adjusting the size function, leading to these geodesics sharing the common limit point with the central geodesic in a specific embedding. We introduce a concept called intertwining, first discussed by Chru\'{s}ciel in \cite{chrusciel2006conformalboundaryextensionslorentzian}. Using an envelopment from Theorem \ref{theorem:endpoint}, we found that the endpoint of the geodesic is a pure singularity if it is not intertwined with any other geodesic in $\mathcal{M}$. The following proposition (Proposition \ref{prop:main} in Section \ref{section 4}) shows that the intertwined geodesic condition provides a clear delineation between a pure singularity and a pure point at infinity.

\begin{prop}\label{prop:main2}
    Let $(\mathcal{M},g,\hat{\mathcal{M}},\psi,\mathcal{C})$ be an envelopment of a maximally extended pseudo-Riemannian manifold $(\mathcal{M},g)$, where $\mathcal{C}$ is the class of geodesics with affine parameter, with $p \in \partial(\psi(\mathcal{M}))$ a pure singularity, and $\gamma: [0,b) \rightarrow \mathcal{M}$ (where $b \in \mathbb{R}^+ \cup \{+\infty\}$) is a non-self-intersecting geodesic in $\mathcal{C}$ without limit points in $\mathcal{M}$. Suppose $\psi(\gamma) \rightarrow p$. By Proposition \ref{prop:3}, there exists another embedding $\psi'$ such that $\psi'(\gamma) \rightarrow q \in \partial(\psi'(\mathcal{M}))$, $p \vartriangleright q$ and no other geodesic in $\mathcal{C}$ ends at $q$. Suppose that no pair of geodesics in $\mathcal{C}$, each of which approaches $q$, is intertwined in $(\mathcal{M},g,\hat{\mathcal{M}}',\psi',\mathcal{C})$. Then $b \in \mathbb{R}^{+}$ (the geodesic $\gamma$ has bounded affine parameter). 
\end{prop}

The important differences between Proposition \ref{prop:main} and the relevant results in \cite{Graf_2024} and \cite{chrusciel2006conformalboundaryextensionslorentzian} are: 
\begin{enumerate}
    \item The technique works for a metric which is at least $C^{2,1}$
    \item The target of the technique is geodesics in a pseudo-Riemannian manifold
    \item An intertwined geodesic is neither a completely global nor local concept. This paper assumes a specific geodesic not intertwined with any other geodesics in $\mathcal{M}$ instead of no intertwined geodesics in $\mathcal{M}$.
\end{enumerate}
Based on Proposition \ref{prop:main2}, we discuss the properties of a pure singularity and its difference from a directional singularity with respect to curve asymptotic behaviour. In Section \ref{section 5} we explore the maximal g-boundary introduced by Graf and Beld-Serrano in \cite{Graf_2024} using the abstract boundary framework. Inspired by the intertwining concept, in Section \ref{section 6} we  test whether the Penrose diagram for the Schwarzschild spacetime is an abstract boundary optimal embedding. This example exemplifies an abstract boundary physical application which connects the purity of the singularity with the curvature singularity, and the inextendibility of the Schwarzschild spacetime (i.e., the curvature singularity in the interior of the Schwarzschild spacetime is a pure singularity, and is $C^0$ inextendible \cite{Sbierski:2015nta}.)\\

The paper is organised as follows. Section \ref{section 2} provides a brief review of the background for the abstract boundary construction. Section \ref{section 3} presents the core application of the Endpoint Theorem by controlling the size of the normal neighbourhood: we show in Proposition \ref{prop:3} that there exists an embedding $\psi$ constructed along a geodesic $\lambda$ without accumulation points in $\mathcal{M}$ such that no other geodesic $\gamma$ in $\mathcal{M}$ ends at the endpoint of $\lambda$ in $\psi$. Section \ref{section 4} delves into the definition of intertwined geodesics under the abstract boundary framework and the structure of pure singularities. We investigate the separability between pure singularities and pure points at infinity if the geodesic without accumulation points in $\mathcal{M}$ is not intertwined with any other geodesics in the newly constructed embedding. In Section \ref{Section Ordering}, we briefly explore an application of Proposition \ref{prop:main} to the existence of a minimal pure singularity, and we further show the application of Proposition \ref{prop:main} to the maximal g-boundary in Section \ref{section 5}. In Section \ref{section 6} the non-intertwined geodesic behaviour of the Schwarzschild spacetime is derived and we show that the Penrose maximal extension of the Schwarzschild spacetime is an abstract boundary optimal embedding \cite{Barry_2014,Barrythesis_2014}. We also examine Wheeler's definition of a black hole and event horizon \cite{Wheeler_2023} in relation to this optimal embedding.\\

\section{Reviewing the Abstract Boundary}\label{section 2}
The core construction of the abstract boundary can be carried out for any pseudo-Riemannian manifold. Unless otherwise specified, the manifolds considered will always be smooth, connected, Hausdorff, paracompact and without boundaries. In all cases, we denote $a$ for a given index number, and $\gamma^a$ means that we consider a specific $a$ component for $\gamma$. Note that $\lambda(t) \rightarrow q$ follows the abstract boundary notation i.e.\ $q$ is the endpoint of the curve $\lambda(t): [0,1) \rightarrow \mathcal{M}$. We denote $\{p\} \vartriangleright \{q\}$ by $p \vartriangleright q$ for convenience. We also adopt the Einstein summation convention in which repeated indices are implicitly summed.

\subsection{The core definitions}
The introduction of the abstract boundary is originally from \cite{scott1994abstract}.
\begin{defn}
A point $p \in \mathcal{M}$ is a \textit{limit point} of a parametrised curve $\gamma: [a,b) \rightarrow \mathcal{M}$ if there is an increasing infinite sequence of real numbers $t_i \rightarrow b$ such that $\gamma(t_i) \rightarrow p$. Moreover, a point $p \in \mathcal{M}$ is an \textit{endpoint} of a parametrised curve $\gamma$ if for every increasing infinite sequence of real numbers $t_i \rightarrow b$, $\gamma(t_i) \rightarrow p$.
\end{defn}

\begin{defn}
An \textit{envelopment} is a triple $(\mathcal{M},\mathcal{\hat{M}},\phi)$ where $\mathcal{M}$ and $\mathcal{\hat{M}}$ are differentiable manifolds of the same dimension $n$ and $\phi$ is an open embedding $\phi: \mathcal{M} \rightarrow \mathcal{\hat{M}}$.
\end{defn}

\begin{defn}
    A \textit{boundary point} $p$ of an envelopment $(\mathcal{M},\mathcal{\hat{M}},\phi)$ is a point $p \in \mathcal{\hat{M}}\backslash\phi(\mathcal{M})$ such that every open neighbourhood $U$ of $p$ in $\mathcal{\hat{M}}$ has non-empty intersection with $\phi(\mathcal{M})$.
\end{defn}

\begin{defn}(Covering Boundary Sets)\\
Let $(\mathcal{M},\mathcal{\hat{M}},\phi,B)$ and $(\mathcal{M},\mathcal{M}',\phi',B')$ be two envelopments for $\mathcal{M}$ with boundary sets $B$ and $B'$ respectively. Then the boundary set $B$ is said to \textit{cover} the boundary set $B'$ if for every open neighbourhood $U$ of $B$ in $\mathcal{\hat{M}}$, there exists an open neighbourhood $U'$ of $B'$ in $\mathcal{M}'$ such that
\begin{align}
    \phi \circ \phi'^{-1}(U' \cap \phi'(\mathcal{M})) \subseteq U .
\end{align}
\end{defn}

\begin{defn}(Equivalent Boundary Sets)\label{Eq}\\
Let $(\mathcal{M},\hat{\mathcal{M}},\phi,B)$ and $(\mathcal{M},\mathcal{M}',\phi',B')$ be two enveloped manifolds with boundary sets $B$ and $B'$ respectively. Then $B$ is said to be \textit{equivalent} to $B'$ $(B \sim B')$ if and only if $B$ covers $B'$ and $B'$ covers $B$.
\end{defn}
\subsection{Classifying abstract boundary points}

\begin{defn}(The a-boundary ${\cal B}(\cal M)$ )\\
If $p$ is a boundary point of an envelopment $(\mathcal{M},\hat{\mathcal{M}},\phi)$, then the equivalence class $[p]$ of the boundary set $\{p\}$ under the equivalence relation $\sim$ of Definition \ref{Eq} is called an abstract boundary point of $\mathcal{M}$. The set ${\cal B}(\cal M)$ of all such abstract boundary points for all envelopments of the given manifold $\mathcal{M}$ is called the \textit{abstract boundary} or the \textit{a-boundary}.
\end{defn}

A class $\mathcal{C}$ of curves in $\mathcal{M}$ satisfies the bounded parameter property if (i) through every point of $\mathcal{M}$ passes at least one curve in $\mathcal{C}$, (ii) every subcurve is a curve in $\mathcal{C}$, (iii) if two curves in $\mathcal{C}$ are related by a change of parameter, then their parameters are either both bounded or are both unbounded. Given an envelopment $(\mathcal{M},\hat{\mathcal{M}},\phi)$, a boundary point $p \in \partial\phi(\mathcal{M})$ is called a $\mathcal{C}$-boundary point (or approachable) if it is a limit point of some curve in $\mathcal{C}$.

\begin{defn}($C^l$ Singular Boundary Point)\\
A boundary point $p$ of an envelopment $(\mathcal{M},g,\mathcal{C},\hat{\mathcal{M}},\phi)$ is said to be \textit{$C^l$ singular} if
\begin{enumerate}
    \item $p$ is not  a $C^l$ regular boundary point,
    \item $p$ is a $\mathcal{C}$-boundary point, and
    \item there exists a curve in the family $\mathcal{C}$ which approaches $p$ with bounded parameter.
\end{enumerate}
\end{defn}

We note that a boundary point $p$ is called a \emph{$C^l$ point at infinity} if it satisfies conditions $1$ and $2$ above with an additional condition that no curve in $\mathcal{C}$ approaches $p$ with bounded parameter.\\

A boundary point $p$ of an envelopment $(\mathcal{M},g,\mathcal{C},\hat{\mathcal{M}},\phi)$ is called a \emph{$C^l$ essential singularity} if it is a $C^l$ singular boundary point which is not covered by a $C^l$ non-singular boundary set $B$ (i.e.\ a boundary set consisting only of points which are all either $C^l$ regular, $C^l$ points at infinity or unapproachable boundary points).\\

We now proceed to the definitions of a directional singularity and a pure singularity which are the core concepts required for the results of this paper.\\

\begin{defn}
An essential singularity $p$ is called a \textit{mixed} or \textit{directional singularity} if $p$ covers a boundary point $q$ which is either regular or a point at infinity. Otherwise, when $p$ covers no such boundary point, we call it a \textit{pure singularity}.
\end{defn}

The following definitions give the key topological relations under investigation.

\begin{defn}(Contact $\bot$).
    Two boundary points $p \in \partial\phi(\mathcal{M})$ and $q \in \partial\phi'(\mathcal{M})$ are \emph{in contact} (denoted $p \bot q$) if for all open neighbourhoods $U$ and $V$ of $p$ and $q$ respectively,
    \begin{equation}\label{contact}
        U \sqcap V := \phi^{-1}(U \cap \phi(\mathcal{M})) \cap \phi'^{-1}(V \cap \phi'(\mathcal{M})) \neq \emptyset  .
    \end{equation}
\end{defn}

\begin{defn}(Separate $\parallel$).\label{defn:separated}
    Two boundary points $p \in \partial\phi(\mathcal{M})$ and $q \in \partial\phi'(\mathcal{M})$ are \emph{separate} (denoted $p \parallel q$) if there exist open neighbourhoods $U$ and $V$ of $p$ and $q$ respectively such that $\phi^{-1}(U \cap \phi(\mathcal{M})) \cap \phi'^{-1}(V \cap \phi'(\mathcal{M})) = \emptyset$ .
\end{defn}

The key concepts used for the research related to the abstract boundary in this paper are the limit points/endpoints of a curve, the cover relation, a pure point at infinity, a pure singularity, and the in contact and separability properties of boundary points.
 
\section{An Application of the Endpoint Theorem}\label{section 3}

The objective of this section is to establish control over the behaviour of geodesics in a neighbourhood of a central, incomplete geodesic $\lambda$. Specifically, our objective is to show that if the curvature along $\lambda$ is bounded in a suitably constructed neighbourhood, then no other geodesic can remain within that neighbourhood and share the same endpoint as $\lambda$. 
In Proposition \ref{prop:3} this result will be formalised, and with use of the Endpoint Theorem, a new envelopment $(\mathcal{M},\hat{\mathcal{M}'},\psi')$ will be constructed such that $\lambda$ has the endpoint $q\in \partial(\psi'(\mathcal{M}))$, and no other geodesic ends at $q$. In the subsequent sections this will enable us to greatly simplify the structure of singular boundary points.

    

Let $(\mathcal{M},g)$ be an $n$-dimensional, pseudo-Riemannian manifold. Consider an affinely parametrised geodesic $\lambda: [0,1) \rightarrow \mathcal{M}$, which is incomplete and has no accumulation points in $\mathcal{M}$. Let $t$ denote the affine parameter along $\lambda$. Consider a frame $\{E_i(0)\}_{0 \leq i \leq n-1}$ at $p = \lambda(0) \in \mathcal{M}$ where $E_0 = \partial/\partial t$ is the tangent vector for $\lambda(t)$ at 0 and the $E_i$, for $1 \leq i \leq n-1$, are vectors spanning the complementary subspace to $E_0$ in $T_p \mathcal{M}$. We parallelly propagate this frame along $\lambda$, ensuring that $\nabla_{\partial/\partial t} E_i = 0$ for all $i$.\\

Following the method of proof of the Endpoint Theorem in \cite{Scott_2021}, we now define a geodesic variation around $\lambda$. Consider the map:
\begin{align*}
    F:(t,x^1,\dots,x^{n-1}) \mapsto \exp_{\lambda(t)}\vast(\sum_{i=1}^{n-1}x^iE_i(t)\vast).
\end{align*}
For a fixed $t$ and fixed unit vector $z \in S^{n-2} \subset \mathbb{R}^{n-1}$, we define a radial geodesic emanating from $\lambda(t)$ by:
\begin{align}
    \sigma(s) = F(t,sz).
\end{align}
Here, $s$ is the radial parameter in the transverse direction. By construction, $\sigma' = z^iE_i$ and $\sigma$ is a geodesic, so $\nabla_{\sigma'}\sigma' = z^iz^j\Gamma_{ij}^{~~k}E_k=0$.\\

There exist (locally bounded) functions $K:[0,1) \rightarrow (0,\infty)$ and $\bar{K}:[0,1) \rightarrow (0,\infty)$ such that, for every $t < 1$ and for every point in the normal neighbourhood of $\lambda(t)$, there exists an $r(t)>0$, such that $\forall~ i,j,k,l,h \in \{1,\dots,n\}$ the following components are bounded:
\begin{align}
    \sup_{t \in [0,1)}\sup_{|x| \leq r(t)}|R_{ijk}^{~~~l}(F(t,x))| \leq \sup_{t \in [0,1)}K(t),\\
    \sup_{t \in [0,1)}\sup_{|x| \leq r(t)}|\partial_{h} R_{ijk}^{~~~l}(F(t,x))| \leq \sup_{t \in [0,1)}\bar{K}(t).
\end{align}

\begin{prop}\label{prop:2}
    Given an $n$-dimensional, pseudo-Riemannian manifold $(\mathcal{M},g)$, and the setting defined above, $\exists$ $ C_0 \in \mathbb{R}$, and $\tilde{r}(t) \leq r(t)$, such that $\forall~i,j,k \in \{1,\dots,n\} $, $ |\Gamma_{ij}^{~~k}(x)| \leq C_0|x|$ for $|x| \leq \tilde{r}(t)$.
\end{prop}
\begin{proof}
The proof proceeds by deriving and analysing the differential equations governing the Christoffel symbols along the transverse directions. The local curvature bound and the Jacobi equation imply a differential inequality for the Christoffel symbols.\\

    We will show that for a sufficiently small $s>0$, there exists $C_0 \in \mathbb{R}$ such that for every $i,j,k \in \{1,\dots,n\} $,  $f(s) = \max_{i,j,k}|\Gamma_{i j}^{~~k}(s)| - C_0s \leq 0$.\label{Gamma bound}\\

The argument proceeds by contradiction. Suppose there exists a radius on which the bound holds, then it implies $\exists ~i, j, k$ such that $\max_{i,j,k}|\Gamma_{i j}^{~~k}| = C_0s$. We take the time derivative of the Jacobi equation, which then leads to an inequality for the time-time component.
\begin{align}
    &\nabla_{t}(\nabla_{\sigma'}\nabla_{\sigma'}\partial_{t} + R(\sigma',\partial_{t})\partial_{\sigma'}) = 0\label{eqn:time-time}\\
    &\nabla_{\sigma'}(\nabla_{t}\partial_{t})|_{s = 0} = - R(\partial_{t},\sigma')\partial_{t}|_{s = 0}\label{ini:time-time}
\end{align}

Equation \ref{ini:time-time} comes from $\nabla_{t}\sigma'|_{s=0} = 0$; the parallel propagation along $\lambda$. Equation \ref{eqn:time-time} reduces to
\begin{align}
    &\nabla_{\sigma'}\nabla_{\sigma'}\nabla_{t}\partial_{t} + (\nabla_{\sigma'}R)(\partial_{t}, \sigma')\partial_{t} + R(\nabla_{\sigma'}\partial_{t}, \sigma')\partial_{t} \nonumber\\
    &+ R(\partial_{t}, \nabla_{\sigma'}\sigma')\partial_{t} + R(\partial_{t}, \sigma')\nabla_{\sigma'}\partial_{t} + (\nabla_{t}R)(\sigma',\partial_{t})\sigma' \nonumber\\
    &+ R(\nabla_{t}\sigma',\partial_{t})\sigma' + R(\sigma',\nabla_{t}\partial_{t})\sigma' + R(\sigma',\partial_{t})\nabla_{t}\sigma' = 0,
\end{align}
\begin{align}
    & |\partial_{s}^2\Gamma_{tt}^{~~i}| = |2z^a(\partial_{s}\Gamma_{tt}^{~~j})\Gamma_{aj}^{~~~i} + \Gamma_{tt}^{~~k}(\partial_{s}z^b\Gamma_{bk}^{~~i}) + z^c z^d\Gamma_{tt}^{~~l}\Gamma_{cl}^{~~~m}\Gamma_{dm}^{~~~i}  \nonumber\\
    &+ z^e z^f R_{e t f~;t}^{~~~~i} + z^g R_{g t t~;s}^{~~~i} + z^h z^n \Gamma_{ht}^{~~\hat{i}} R_{n \hat{i} t}^{~~~i} + z^oz^p\Gamma_{o t}^{~~i'}R_{p t {i}^{'}}^{~~~~i} \nonumber\\
    &+ z^q \Gamma_{t q}^{~~j'}R_{{j}^{'} t t}^{~~~i} + z^u z^v \Gamma_{t t}^{~~k'}R_{u {k}^{'} v}^{~~~i} + z^w z^x\Gamma_{t w}^{~~\bar{i}}R_{x t \bar{i}}^{~~~i}| \nonumber\\
    & \leq (n-2)^2\bar{K} + (n-2)\bar{K} + 2(n-2)C_0^2s + (n-2)C_0^2s  \nonumber\\
    &  + 4(n-2)^2KC_0s + (n-2)KC_0s + (n-2)^2(C_0s)^3 \label{ineq:time-time}.
\end{align}
Along $\lambda$, we have $|\Gamma_{t t}^{~~i}(s)|_{s=0} = 0$, $|\partial_s\Gamma_{t t}^{~~i}(s)|_{s=0} = K$ and inequality \ref{ineq:time-time} as $s \rightarrow 0$, so $|\Gamma_{tt}^{~~i}(s)| \leq Ks + \tilde{C}s^2 + \tilde{C}'s^3 \leq C_0s$ for some $C_0$ and a sufficiently small $s$, where $\tilde{C} = 2(n-2)^2\bar{K} + 2(n-2)\bar{K}$ and $\tilde{C}' = 2(n-2)C_0^2 + (n-2)C_0^2 + 4(n-2)^2KC_0 + (n-2)KC_0$.\\

Now we consider the angular-angular component and the time-angular component along the radial geodesic $\sigma$.\\

Define
\begin{align*}
    E_i|_{F(t,x^1,\dots,x^{n-1})} = \frac{\partial}{\partial x^i}F(t,x^1,\dots,x^{n-1}).
\end{align*}
We build the neighbourhood along the radial geodesics from the central geodesic $\sigma: r \rightarrow F(t,rz)$.
\begin{align*}
    J_{\alpha}(r) = \frac{\partial}{\partial \epsilon}F(r(z + \epsilon E_{\alpha})) = rE_{\alpha}.
\end{align*}
Note that $s$ is the parameter on a single geodesic, and $r$ is the parameter on every radial geodesic simultaneously (varying $\epsilon$ changes the direction, but $r$ still measures the radial distance from $\lambda(t)$). We keep $s$ in the rest of the calculations for convenience.\\

We now focus on the transverse components:
\begin{align}
    &\nabla_{t}(\nabla_{\sigma'}\nabla_{\sigma'}J_{\alpha} + R(\sigma',J_{\alpha})\partial_{\sigma'}) = 0\label{eqn:time-angular}\\
    &\nabla_{\sigma'}(\nabla_{t}J_{\alpha})|_{s = 0} = 0\label{ini:time-angular}\\
    &\nabla_{\alpha}(\nabla_{\sigma'}\nabla_{\sigma'}J_{\beta} + R(\sigma',J_{\beta})\sigma') = 0\label{eqn:angular-angular}\\
    &\nabla_{\sigma'}(\nabla_{\alpha}J_{\beta})|_{s = 0} = 0.\label{ini:angular-angular}
\end{align}
Equation \ref{ini:time-angular} comes from $\nabla_{s}E_{\alpha}|_{s=0} = 0$ because $0 = \nabla_{\sigma'}\sigma' = z^iz^j\Gamma_{i j}^{~~k}E_k$ due to the parallel propagation along $\sigma$. For Equation \ref{ini:angular-angular}, we start from the Jacobi equation $\nabla_{\sigma'}\nabla_{\sigma'}J_{\alpha} + R(J_{\alpha},\sigma')\sigma' = 0$. The equation reduces to
\begin{align*}
    \nabla_{\sigma'}\nabla_{\sigma'}E_{\gamma} + \frac{2}{s}\nabla_{\sigma'}E_{\gamma} + R(E_{\gamma},\sigma')\sigma' = 0\\
    3\nabla_{\sigma'}\nabla_{\sigma'}E_{\gamma} + R(E_{\gamma},\sigma')\sigma' = 0~~\text{as}~s \rightarrow 0\\
    3z^{\alpha}z^{\beta}\nabla_{E_{\alpha}}\nabla_{E_{\beta}}E_{\gamma} + R(E_{\gamma},\sigma')\sigma' = 0\\
    \nabla_{(E_{\alpha}}\nabla_{E_{\beta})}E_{\gamma} = - C R(E_{\gamma},\sigma')\sigma',
\end{align*}
where $C \in \mathbb{R}$. Since we know that the anti-symmetric part is $R(E_{\alpha},E_{\beta})E_{\gamma}$, then we have
\begin{align*}
    \nabla_{E_{\alpha}}\nabla_{E_{\beta}}E_{\gamma} = - CR(E_{\gamma},\sigma')\sigma'.
\end{align*}
Thus, by multiplying the above equation by $z^{\alpha}$, we have
\begin{align*}
    \nabla_{\sigma'}\nabla_{E_{\beta}}E_{\gamma} = - CR(E_{\gamma},\sigma')\sigma'.
\end{align*}
We now return to Equations \ref{eqn:time-angular} and \ref{eqn:angular-angular}, which reduce to the following forms respectively:
\begin{align}
    2\nabla_{t}\nabla_{\sigma'}E_{\beta} + s\nabla_{t}\nabla_{\sigma'}\nabla_{\sigma'}E_{\beta} + s\nabla_{t}R(\sigma',E_{\beta})\sigma' = 0
\end{align}
\begin{align}
    2s\nabla_{E_{\beta}}\nabla_{\sigma'}E_{\gamma} + s^2\nabla_{E_{\beta}}\nabla_{\sigma'}\nabla_{\sigma'}E_{\gamma} + s^2\nabla_{E_{\beta}}R(\sigma',E_{\gamma})\sigma' = 0.
\end{align}
Following a similar calculation as for the time-time component with a limit of $s \rightarrow 0$, we have
\begin{align}
    |\partial_{s}^2\Gamma_{\beta t}^{~~~i}| \leq \tilde{C} + \tilde{C}'s\label{ineq:time-angular},
\end{align}
\begin{align}
    |\partial_{s}^2\Gamma_{\beta \gamma}^{~~~~i}| \leq \tilde{C} + \tilde{C}'s.\label{ineq:angular-angular}
\end{align}

Thus, following the same argument as used for $\Gamma_{tt}^{~~i}$, we have
\begin{align*}
    |\Gamma_{\beta t}^{~~~i}| < C_0s,~~ |\Gamma_{\beta \gamma}^{~~~~i}| < C_0s.
\end{align*}
Suppose $\max_{i,j,k}|\Gamma_{i j}^{~~k}(s)|$ fails to be bounded by $C_0s$ at some $s \neq 0$ in a sufficiently small $s$ region given in a time interval, then $\max_{i,j,k}|\partial_s\Gamma_{i j}^{~~k}(s)| - C_0$ should be non-negative. For inequalities \ref{ineq:time-time}, \ref{ineq:time-angular}, and \ref{ineq:angular-angular}, their first derivative is strictly negative if $\tilde{r}(t) < C_0/\tilde{C}$, i.e.\ there exists a sufficiently small $\tilde{r}(t)$ and $C_0 \in \mathbb{R}$ which contradicts that $\max_{i,j,k}|\partial_s\Gamma_{i j}^{~~k}(s)| - C_0 \geq 0$. \\

Since $z$ has unit length and $s$ is the geodesic parameter, the norm of the transverse coordinates is $|x| = s$. This completes the proof.

\end{proof}
\begin{prop}\label{prop:3}
   Given an $n$-dimensional, pseudo-Riemannian manifold $(\mathcal{M},g)$ and an envelopment $(\mathcal{M},\hat{\mathcal{M}},\psi)$, if $\lambda$ is an affinely parametrised geodesic in $\mathcal{M}$ and $\psi(\lambda) \rightarrow p \in \partial(\psi(\mathcal{M}))$, then there always exists an envelopment $(\mathcal{M},\hat{\mathcal{M}}',\psi')$ such that $\psi'(\lambda) \rightarrow q \in \partial(\psi'(\mathcal{M}))$, there exists no other geodesic $\gamma$ in $\mathcal{M}$ such that $\psi'(\gamma) $ ends at $ q $, and $p \vartriangleright q$. 
\end{prop}

\begin{proof}

Using the local coordinates $(t,x^1,\dots,x^{n-1})$ defined above, we can write the geodesic equations for any $0 \leq i,j \leq n-1$ and $1 \leq a \leq n-1$ as 
\begin{align}
    \frac{d^2 x^a}{d\tau^2} + \Gamma^{~~a}_{ij}\frac{dx^i}{d\tau}\frac{d x^j}{d\tau} = 0 . \label{eq:40}
\end{align}
Rearranging the second derivative of $x^a$ in terms of $t$, we have
\begin{align}
    \frac{d^2 x^a}{dt^2} = \frac{1}{\dot{t}^2}\vast(-\frac{\ddot{t}}{\dot{t}}\dot{x}^a + \ddot{x}^a\vast) \label{eq:42}
\end{align}
where $\dot{x}^a$ is the derivative in terms of the affine parameter $\tau$. \\

We then obtain the equation for any $1 \leq a \leq n-1$ to control the nearby geodesic behaviour. $d^2x^a/dt^2$ for each $x^a$ is given by:
\begin{align}
    \frac{d^2x^a}{dt^2} =&~\Gamma^{~~t}_{ij}\frac{\dot{{x}}^a}{\dot{t}}\frac{\dot{{x}}^i}{\dot{t}}\frac{\dot{x}^j}{\dot{t}} - \Gamma^{~~a}_{i j}\frac{\dot{{x}}^i}{\dot{t}}\frac{\dot{{x}}^j}{\dot{t}} \nonumber\\
    =& - \Gamma^{~~a}_{t t} + (\Gamma^{~~t}_{t t}-2\Gamma^{~~a}_{a t})\frac{\dot{{x}}^a}{\dot{t}} - 2\Gamma_{t b}^{~~a}\frac{\dot{{x}}^b}{\dot{t}} \nonumber\\
    -& (\Gamma^{~~a}_{a a} - 2\Gamma^{~~t}_{{a} t})\vast(\frac{\dot{{x}}^a}{\dot{t}}\vast)^2 - \Gamma^{~~a}_{b c}\frac{\dot{{x}}^b}{\dot{t}}\frac{\dot{{x}}^c}{\dot{t}} - 2\Gamma_{a b}^{~~a} \frac{\dot{{x}}^a}{\dot{t}}\frac{\dot{{x}}^b}{\dot{t}} + 2\Gamma_{t b}^{~~t}\frac{\dot{{x}}^a}{\dot{t}}\frac{\dot{{x}}^b}{\dot{t}} \nonumber\\
    + & 2\Gamma^{~~t}_{a b}\vast(\frac{\dot{{x}}^a}{\dot{t}}\vast)^2\frac{\dot{{x}}^b}{\dot{t}} + \Gamma^{~~t}_{b c}\frac{\dot{{x}}^a}{\dot{t}}\frac{\dot{{x}}^b}{\dot{t}}\frac{\dot{{x}}^c}{\dot{t}} + \Gamma^{~~~t}_{a a}\vast(\frac{\dot{{x}}^a}{\dot{t}}\vast)^3   \label{eq:43}
\end{align}
where $ c \neq a \neq 0$ and $b \neq a \neq 0$. Using Proposition \ref{prop:2}, Equation \ref{eq:43} can be simplified to:
\begin{align}\label{eq:431}
    \vast|\frac{d^2x^a}{dt^2}\vast| \leq & |\Gamma^{~~a}_{tt}| + (|\Gamma^{~~t}_{tt}| + 2|\Gamma^{~~a}_{a t}|)\vast|\frac{\dot{x}^a}{\dot{t}}\vast| + 2|\Gamma_{t b}^{~~a}|\vast|\frac{\dot{x}^b}{\dot{t}}\vast| \nonumber \\
    + & (2|\Gamma^{~~t}_{a t}| + |\Gamma^{~~a}_{{a}{a}}|)\vast|\frac{\dot{x}^a}{\dot{t}}\vast|^2 + 2|\Gamma^{~~t}_{t b}|\vast|\frac{\dot{{x}}^a}{\dot{t}}\vast|\vast|\frac{\dot{x}^b}{\dot{t}}\vast|  + |\Gamma^{~~a}_{b c}|\vast|\frac{\dot{{x}}^b}{\dot{t}}\vast|\vast|\frac{\dot{{x}}^c}{\dot{t}}\vast| \nonumber\\
    + & 2|\Gamma_{a b}^{~~a}| \vast|\frac{\dot{{x}}^a}{\dot{t}}\vast|\vast|\frac{\dot{{x}}^b}{\dot{t}}\vast| + 2|\Gamma_{t b}^{~~t}|\vast|\frac{\dot{{x}}^a}{\dot{t}}\vast|\vast|\frac{\dot{{x}}^b}{\dot{t}}\vast| + |\Gamma^{~~t}_{a a}|\vast|\frac{\dot{x}^a}{\dot{t}}\vast|^3 + 2|\Gamma^{~~t}_{a b}|\vast|\frac{\dot{{x}}^a}{\dot{t}}\vast|^2\vast|\frac{\dot{x}^b}{\dot{t}}\vast| \nonumber\\
    + & |\Gamma^{~~t}_{b c}|\vast|\frac{\dot{{x}}^a}{\dot{t}}\vast|\vast|\frac{\dot{{x}}^b}{\dot{t}}\vast|\vast|\frac{\dot{{x}}^c}{\dot{t}}\vast|\nonumber \\
    \leq& C_0s + (2n-1)C_0s\vast|\frac{\dot{{x}}}{\dot{t}}\vast| +  (n^2 - 1)C_0s\vast|\frac{\dot{{x}}}{\dot{t}}\vast|^2 + (n-1)^2C_0s\vast|\frac{\dot{{x}}}{\dot{t}}\vast|^3.
\end{align}
We can express $\sum_{a=1}^{n-1}|d^2x^a/dt^2|$ using the inequality \ref{eq:431}:
\begin{align}\label{ineq:4}
    \sum_{a=1}^{n-1}\vast|\frac{d^2x^a}{dt^2}\vast| \leq& (n-1)C_0s + (n-1)(2n-1)C_0s\vast|\frac{\dot{{x}}}{\dot{t}}\vast|  \nonumber\\
    + & (n-1)(n^2-1)C_0s\vast|\frac{\dot{{x}}}{\dot{t}}\vast|^2 + (n-1)^3C_0s\vast|\frac{\dot{{x}}}{\dot{t}}\vast|^3 .
\end{align}
Let $q = \sqrt{\sum_{a = 1}^{n-1}(|x^a|^2 +  |dx^a/dt|^2)}$. We then substitute the inequality \ref{ineq:4} into the following equation:
\begin{align}
     2\dot{q}q &= \frac{d}{dt}q^2 = \sum_{a=1}^{n-1}\frac{d}{dt}\vast(|x^a|^2 + \vast|\frac{dx^a}{dt}\vast|^2\vast) = \sum_{a=1}^{n-1}2|x^a|\vast|\frac{dx^a}{dt}\vast| + 2\vast|\frac{dx^a}{dt}\vast|\vast|\frac{d^2x^a}{dt^2}\vast| \nonumber \\
    & \leq 2(n-1)|x|\vast|\frac{dx}{dt}\vast| + 2(n-1)\vast|\frac{dx}{dt}\vast|\vast((n-1)C_0s + (n-1)(2n-1)C_0s\vast|\frac{dx}{dt}\vast|  \nonumber\\
    & +  (n-1)(n^2-1)C_0s\vast|\frac{dx}{dt}\vast|^2 + (n-1)^3C_0s\vast|\frac{dx}{dt}\vast|^3\vast). \label{eq:44}
\end{align}
For any $a$ such that $2|x^{a}||dx^{a}/dt| \leq q^2$, $|dx^{a}/dt|^2 \leq q^2$, $|x^{a}|^2 \leq q^2$, these are substituted into equation \ref{eq:44}:
\begin{align}
    \frac{dq}{dt}& \leq \frac{1}{2}n(n-1)C_0q + \frac{1}{2}(n-1)^2(2n-1)C_0q^2  \nonumber\\
    &+ \frac{1}{2}(n-1)^2(n^2-1)C_0q^3 + \frac{1}{2}(n-1)^4C_0q^4\label{eq:45},
\end{align}
\begin{align}
    \frac{dq}{dt} < \frac{1}{2}n(n-1)C_0q ~~(\text{when}~|q| < 1).
\end{align}
This induces a bound for $|x^a|$ and $|dx^a/dt|$:
\begin{align}
    |x^a| < q(0)e^{\bar{C}(t)}\label{eq:47}\\
    \vast|\frac{dx^a}{dt}\vast| < q(0)e^{\bar{C}(t)},\label{eq:48}
\end{align}
where $\bar{C}(t) = 1/2n(n-1)C_0t$. \\

We now analyse the asymptotic behaviour as $t \rightarrow 1$ under the assumption of bounded Christoffel symbols. Specifically, we have $|\Gamma_{ij}^{~~k}| \leq C_0s$ on $[0,t] \times \{|sz| \leq f_{+}(t)\} $ for any $t \in I = [0,1)$ where $f_+$ is a function $f_+: [0,1) \rightarrow \mathbb{R}^+$. From our previous bounds inequality \ref{ineq:4}, we have:
\begin{align}
    |(x^i )''| < (n-1)C_0(|x^i| + 3|x^i||dx^i/dt|)
\end{align}
where $|x^i| < C_0$ and $(x^i)'$ is the first derivative with respect to $t$. By the Taylor series expansion, we get 
\begin{align*}
    &|x^i(t+h) - x^i(t) - (x^i)'(t)|I|| \leq \frac{1}{2}\sup_{t \in I}\vast|\frac{d^2x^i}{dt^2}\vast||I|^2\\
    &\vast|\frac{dx^i}{dt}\vast| \leq \frac{\sup_{t \in I}|x^i|}{|I|} + \frac{1}{2}\sup_{t \in I}\vast|\frac{d^2x^i}{dt^2}\vast||I|\\
    &\ \ \ \ \ \ \ \leq \frac{\sup_{t \in I}|x^i|}{|I|} + \frac{1}{2}(n-1)C_0\vast(\sup_{t \in I}|x^i||I| + \sup_{t \in I}|x^i|\vast|\frac{dx^i}{dt}\vast||I|\vast)\\
    &\vast(1 - \frac{1}{2}(n-1)C_0|I|\sup_{t \in I}|x^i|\vast)\vast|\frac{dx^i}{dt}\vast| \leq \vast(\frac{1}{|I|} + \frac{1}{2}(n-1)C_0|I|\vast)\sup_{t \in I}|x^i|\\
    &\vast|\frac{dx^i}{dt}\vast|\leq \frac{\sup_{t \in I}|x^i|}{1 - 1/2(n-1)C_0|I|\sup_{t \in I}|x^i|} \vast(\frac{1}{|I|} + \frac{1}{2}(n-1)C_0|I|\vast).
\end{align*}
With a suitable choice of $f_+(t) = C(1-t)^2$ and a small enough interval $I = [1 - \epsilon,1)$, $\sup_{t \in I}|x^i|/\epsilon \leq f_+(t)/\epsilon \leq C\epsilon$ on $I$. We then have that $|dx^i/dt| \rightarrow 0$ as $t \rightarrow 1$.\\

Now we return to $q = {|x^i |}^2 + {|(x^i)'|}^2$. We have that $q(t) \rightarrow 0$ as $t \rightarrow 1$ from the above inequality.
\begin{align*}
    q(0) = q(t)e^{\int_0^{t}\bar{C}(\tau) d\tau}
\end{align*}
We then obtain
\begin{align*}
    q(0) \leq q(t)e^{\frac{1}{2}C't^2} \rightarrow q(1)e^{\frac{1}{2}C'} \rightarrow 0~~\text{as}~t \rightarrow 1.
\end{align*}
This is because $e^{\int_0^{t}\bar{C} (\tau) d\tau}$ is bounded from the bounded curvature asymptotic behaviour, so $q(0) = 0$ which implies that the geodesic within $f_+(t)$ is the central geodesic.\\

We now address the more complicated case where the curvature may become unbounded as $t \rightarrow 1$. In this case, we cannot rely on the local bound, but we can still establish convergence through a refined barrier function approach.\\

Consider the behaviour on intervals $[t - \delta(t), t]$ following the previous argument given the bound on $|\Gamma_{ij}^{~~k}|$ as $t \rightarrow 1$. Using the interpolation inequality, we bound the first derivatives in terms of the function values and the second derivatives:
\begin{align*}
    \sup_{t \in I} |(x^i)'| &\leq a\sup_{t \in I}|x^i|^{\frac{1}{2}}\sup_{t \in I}|(x^i)''|^{\frac{1}{2}}\\
    & \leq a\sup_{t \in I}|x^i|^{\frac{1}{2}}\vast((n-1)(C_0K(t))\vast(\frac{|x^i|}{\delta} + 3|x^i||(x^i)'|\vast)\vast)^{\frac{1}{2}}\\
    &\leq a\sup_{t \in I}|x^i|((n-1)(C_0K(t)))^\frac{1}{2}\vast(\frac{1}{\delta} + 3|(x^i)'|\vast)^{\frac{1}{2}}\\
    & \leq af_+(t-\delta)((n-1)(C_0K(t)))^\frac{1}{2}\vast(\frac{1}{\delta} + 3|(x^i)'|\vast)^{\frac{1}{2}}\\
    &\leq af_+(t-\delta)((n-1)(C_0K(t)))^\frac{1}{2}\vast(\frac{1}{\delta} + 6|(x^i)'|\vast)
\end{align*}
where $a \in \mathbb{R}$. Because $f_+(t)$ should be small, we choose $af_+(t-\delta)((n-1)(C_0K(t)))^\frac{1}{2} < 1/12$.\\

The above inequality holds for a sufficiently large $K(t)$ as $t \rightarrow 1$. Thus, we have
\begin{align*}
     \sup_{t \in I} |(x^i)'| \leq 2af_+(t-\delta(t))((n-1)(C_0K(t)))^\frac{1}{2}\frac{1}{\delta}.
\end{align*}

We choose 
\begin{align}\label{f(t)}
    f_+^2(\tau) = \delta^2(t)(1-t)^2/M(t),
\end{align} 
where $\tau = t - \delta(t)$, $M(t) = \bar{C}'(t)\exp(\int_0^t\bar{C}(\tau)d\tau)$ and $\delta(t) = 1 - t$ for any $t \in (0,t]$:
\begin{align*}
    q(t) < \vast(\frac{f_+(t-\delta(t))}{\delta(t)}\bar{C}'(t)\vast)^2 + f_+^2(t)
\end{align*}
where $\bar{C}'(t) = 2a((n-1)(C_0K(t)))^\frac{1}{2}$.
\begin{align*}
    q(0) \leq q(t)e^{\int_0^tC'(\tau)d\tau} \leq \vast(\frac{f_+(t-\delta(t))}{\delta(t)}\vast)^2M(t) = (1-t)^2 \rightarrow 0~~\text{as}~t \rightarrow 1.
\end{align*}

In the above setting, $q({t}) \rightarrow 0$ as $t \rightarrow 1$ such that $q(0) \rightarrow 0$ by making $q(t_k)e^{\int_0^{t}(C_0K(t))^{\frac{1}{2}}\bar{C}tdt} \rightarrow 0$. This implies that the only geodesic that satisfies this condition is the central geodesic.\\

Following the previous discussion, in both cases, the choice of $f_+(t)$ forms a new neighbourhood inside $F(t,sz)$. The function $f_+(t)$ here is equivalent to the function $f(t)$ given in the proof of the Endpoint Theorem \cite{Scott_2021}. Therefore, by the choice of $f_{+}$ above, we construct the charts $\mu$ and $\mu'$, with which a new embedding $\psi': \mathcal{M} \rightarrow \hat{\mathcal{M}'}$ is formed such that $\psi'(\lambda) \rightarrow q \in \partial\psi'(\mathcal{M})$, and for which no other geodesic $\psi'(\gamma(t))$ will stay within the normal neighbourhood as $t \rightarrow 1$, i.e. $\lambda$ is the only geodesic that enters and remains within this neighbourhood of $q \in \partial\psi'(\mathcal{M})$.\\

Now let $\{x_i\}$ be a sequence of points in $\mathcal{M}$ such that $\psi'(x_i) \rightarrow q$. There exists an open neighbourhood $U$ of $q$ in $\hat{\mathcal{M}'}$ such that $$ U \cap \psi'(\mathcal{M}) \subset \psi'(\mu).$$ Since $\psi'(x_i) \rightarrow q$, for some $j$, the sequence $\psi'(x_i) \in \psi'(\mu)$ for all $i \geq j$. This means that $\psi(x_i) \in \psi(\mu)$ for all $i \geq j$. Since the sequence $\{\psi(x_i)\}$ must approach the boundary of $\psi(\mathcal{M})$, and for the normal neighbourhood of $\lambda$ used to define the chart $\mu$ the function $f_{+}(t) \rightarrow 0$ as $t \rightarrow 1$, it must be the case that $\psi(x_i) \rightarrow p$. So from Theorem 19 of \cite{scott1994abstract}, $p \vartriangleright q$.



\end{proof}

Proposition \ref{prop:3} holds if the metric is at least $C^{2,1}$, since the connection in the Jacobi equation involves the derivative of the Riemann tensor. Lower regularity cases than this change the normal neighbourhood construction in the proof, so there is no simple method to produce the same conclusion with a lower regularity than $C^{2}$.\\

Proposition \ref{prop:3} is also true for complete geodesics and just needs a minor change of the proof where it depends on the finiteness of the affine parameter of the geodesic. This relates to the choice of the barrier function $f_+(t)$, but this does not fundamentally change the proof, so the result for a complete geodesic also prevents other geodesics from having the same endpoint.

\section{Properties of a Pure Singularity}\label{section 4}
From Proposition \ref{prop:3}, a special case which is immediately apparent is a pseudo-Riemannian manifold without intertwined geodesics. This concept was first introduced by Piotr Chru\'{s}ciel in \cite{chrusciel2006conformalboundaryextensionslorentzian} to explain the null geodesic behaviour in the Taub-NUT spacetime. In the abstract boundary context, we also describe null geodesics in the Misner spacetime via their limiting behaviour (not necessarily terminating at an endpoint). In this section we define curve limiting behaviour in a fixed envelopment and compare it with the intertwined geodesics introduced by Graf and Beld-Serrano \cite{Graf_2024}. This framework enables us to analyse the topological structure of a pure singularity via the position of the associated embedding constructed from Proposition~\ref{prop:3} in the cover relation chain.\\

We denote by $\pi^{T\mathcal{M}}:T\mathcal{M} \rightarrow \mathcal{M}$ the natural projection map. We also fix a complete Riemannian background metric $h^{T\mathcal{M}}$ on $T\mathcal{M}$. Given a fixed $X \in T\mathcal{M}$ and $r>0$, let $B_r(X)$ denote the open ball of radius $r$ in $T\mathcal{M}$ around $X$. Moreover, for any $X \in T\mathcal{M}$, let $\gamma_{X}:(a_X,b_X) \rightarrow \mathcal{M}$ be the unique inextendible geodesic in $\mathcal{M}$ with initial data $\gamma_X(0) = \pi^{T\mathcal{M}}(X)$, $\dot{\gamma}_{X}(0) = X$. Note that $X \rightarrow a_X$ is upper semi-continuous and $X \rightarrow b_X$ is lower semi-continuous.
\begin{defn}(Thickening).
    Let $(\mathcal{M},g,\hat{\mathcal{M}},\psi,\mathcal{C})$ be an envelopment of a pseudo-Riemannian manifold $(\mathcal{M},g)$. For $X \in T\mathcal{M}$ and $r>0$ the thickening of radius $r$ generated from $X$ is
    \begin{align}\label{Thickening}
        O_{X,r} \equiv O^{\partial}_{X,r} \cup O^{int}_{X,r}
    \end{align}
    where the boundary thickening $O^{\partial}_{X,r}$ and the interior thickening $O^{int}_{X,r}$ are defined as follows:
    \begin{align}
        O^{int}_{X,r} \equiv = \{(\psi \circ \gamma_{Y})((0,b_{Y})):Y \in B_r(X) \}
    \end{align}
    and 
    \begin{align}
        O^{\partial}_{X,r} \equiv = \{\lim_{t \rightarrow b^{-}_Y}(\psi \circ \gamma_{Y})((0,b_{Y})):Y \in B_r(X) ~\text{s.t.\ this limit exists in}~\hat{\mathcal{M}}\}.
    \end{align}
\end{defn}

The following definition of intertwined geodesics is that given by Graf and Beld-Serrano \cite{Graf_2024}, but we have generalised their version which focussed on timelike geodesics to extend it to the class of all geodesics.

\begin{defn}
    (Intertwined Geodesics)\label{defn:intertwined}. Let $\gamma :[0,b_{Y_1}) \rightarrow \mathcal{M}$, $Y_1 := \dot{\gamma}(0)$ and $\gamma':[0,b_{Y_2}) \rightarrow \mathcal{M}$, $Y_2 := \dot{\gamma}'(0)$, be two inextendible geodesics without limit points in an at least $C^2$ pseudo-Riemannian manifold $(\mathcal{M}, g)$. Then, we say that $\gamma$ and $\gamma'$ are not intertwined provided that one of the following conditions holds:
    \begin{enumerate}
        \item For any radii $r > 0, \rho > 0$ there exist $s_1 \in (0, b_{Y_1})$, $s_2 \in (0, b_{Y_2})$ such that $\gamma([s_1, b_{Y_1})) \subset O^{\mathcal{M}}_{Y_2,\rho}$ and $\gamma'([s_2, b_{Y_2})) \subset O^{\mathcal{M}}_{Y_1,r}$,
        \item There exists $s_1 \in (0, b_{Y_1})$, $ s_2 \in (0, b_{Y_2})$ and radii $r, \rho > 0$ such that $O^{\mathcal{M}}_{\dot{\gamma}(s_1),r} \cap O^{\mathcal{M}}_{\dot{\gamma}'(s_2),\rho} = \emptyset$.
    \end{enumerate}
    If neither of these conditions hold, then we say that $\gamma$ and $\gamma'$ are {\it intertwined}.
\end{defn}

Under this definition, for two geodesics to be intertwined requires $O^{\mathcal{M}}_{\dot{\gamma}(s_1),r} \cap O^{\mathcal{M}}_{\dot{\gamma}'(s_2),\rho} \neq \emptyset$. This splits into two cases: 1. The two intertwined geodesics intersect infinitely many times. 2. The thickenings of the two geodesics intersect, but the geodesics themselves do not intersect infinitely many times.\\


We now introduce a new definition of intertwined geodesics which is well-suited to the abstract boundary framework and relates to envelopments of a pseudo-Riemannian manifold. It enlarges the class of curves under consideration from the class of geodesics to a class $\mathcal{C}$ of curves in $\mathcal{M}$ which satisfies the bounded parameter property.

\begin{defn}(Intertwined Curves).\label{defn:Sintertwined}
    Let $(\mathcal{M},g,\hat{\mathcal{M}},\psi,\mathcal{C})$ be an envelopment of a pseudo-Riemannian manifold  $(\mathcal{M},g,\mathcal{C})$. We know that $\hat{\mathcal{M}}$ admits a complete Riemannian metric $\hat{h}$ which induces a complete distance $d_{\hat{h}}$ on $\hat{\mathcal{M}}$. Consider the two curves in $\mathcal{C}$, $\gamma:[0,b_{\gamma}) \rightarrow \mathcal{M}$ and $\gamma':[0,b_{\gamma'}) \rightarrow \mathcal{M}$ which are inextendible curves without limit points in $\mathcal{M}$. We say that $\gamma$ and $\gamma'$ are {\it intertwined curves} if:
    \begin{enumerate}
        \item there exists an increasing infinite sequence of real numbers $\{x_i\}_{i \in \mathbb{N}}$ in $[0,b_{\gamma})$ with $x_i \rightarrow b_{\gamma}$ as $i \rightarrow \infty$ and an increasing infinite sequence of real numbers $\{y_i\}_{i \in \mathbb{N}}$ in $[0,b_{\gamma'})$ with $y_i \rightarrow b_{\gamma'}$ as $i \rightarrow \infty$ such that $d_{\hat{h}}(\psi(\gamma(x_i)),\psi(\gamma'(y_i))) \rightarrow 0$ as $i \rightarrow \infty$, and;
        \item there exists an increasing infinite sequence of real numbers $\{x_i\}_{i \in \mathbb{N}}$ in $[0,b_{\gamma})$ with $x_i \rightarrow b_{\gamma}$ as $i \rightarrow \infty$ and an increasing infinite sequence of real numbers $\{y_i\}_{i \in \mathbb{N}}$ in $[0,b_{\gamma'})$ with $y_i \rightarrow b_{\gamma'}$ as $i \rightarrow \infty$ such that $d_{\hat{h}}(\psi(\gamma(x_i)),\psi(\gamma'(y_i)))$ has a lower bound of $l$ as $i \rightarrow \infty$, where $l \in \mathbb{R}^+ \cup \{\infty\}$.
    \end{enumerate}
\end{defn}

We note that Definition \ref{defn:intertwined} is an embedding-independent concept, whereas Definition \ref{defn:Sintertwined} is designed around the limiting behaviour of the curves, and so it is envelopment-dependent -- it effectively provides a topological constraint on the choice of curves in $\mathcal{C}$. See Figure \ref{fig1} which depicts two different envelopments of the Misner spacetime, and illustrates how a pair of geodesics can be not intertwined in one envelopment but intertwined in another envelopment in accordance with Definition \ref{defn:Sintertwined}, but their classification remains unchanged using Definition \ref{defn:intertwined}. With Definition \ref{defn:Sintertwined}, for the envelopment $\psi$, the classification of the pair of curves as intertwined or not intertwined is invariant under the choice of complete Riemannian metric $\hat{h}$ on $\hat{\mathcal{M}}$.\\

\begin{figure}[htbp]
\centerline{\includegraphics[scale=0.6]{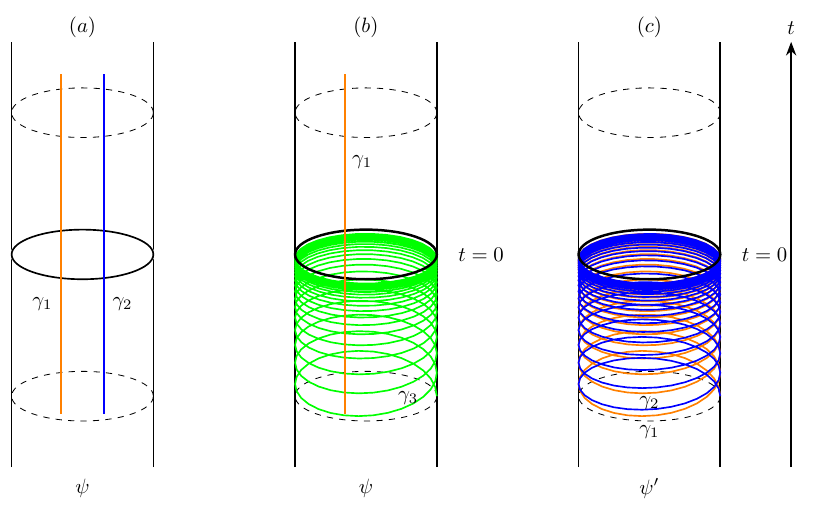}}
\caption{This diagram, in three parts, shows the difference between Graf and Beld-Serrano's intertwined geodesic concept and our new intertwined curve concept, using two particular envelopments $\psi$ and $\psi'$ of the Misner spacetime ($t<0$), in each case considering two specific geodesics as they approach $t=0$. $(a)$ depicts two geodesics (both straight) in the envelopment $\psi$ which are not intertwined by both Definition \ref{defn:intertwined} and Definition \ref{defn:Sintertwined}, $(b)$ shows a pair of geodesics (one straight and one spiralling) in the envelopment $\psi$ which are intertwined by both Definition \ref{defn:intertwined} and Definition \ref{defn:Sintertwined}, and $(c)$ depicts the two geodesics $\gamma_1$ and $\gamma_2$ which are straight in the envelopment $\psi$ (see $(a)$) but have become spiralling in the envelopment $\psi'$, and they are intertwined by Definition \ref{defn:Sintertwined} but not intertwined by Definition \ref{defn:intertwined}.}
\label{fig1}
\end{figure}

The new definition can be readily applied to categorise how pairs of inextendible curves without limit points in $\mathcal{M}$ approach the boundary in an envelopment of a pseudo-Riemannian manifold  $(\mathcal{M},g,\mathcal{C})$. The following proposition considers the often encountered problematic case where one curve approaches a boundary point $p$ as an endpoint, and the other curve also approaches $p$, but not as an endpoint. 
\begin{prop}\label{prop:intertwined} 
    Let $(\mathcal{M},g,\hat{\mathcal{M}},\psi,\mathcal{C})$ be an envelopment of a pseudo-Riemannian manifold  $(\mathcal{M},g,\mathcal{C})$. Consider the two curves in $\mathcal{C}$, $\gamma:[0,b_{\gamma}) \rightarrow \mathcal{M}$ and $\gamma':[0,b_{\gamma'}) \rightarrow \mathcal{M}$ which are inextendible curves without limit points in $\mathcal{M}$. If $\psi(\gamma(t)) \rightarrow p \in  \partial\psi(\mathcal{M})$ as $t \rightarrow b_{\gamma}$ and $\psi(\gamma'(t))$ approaches $p$ as $t \rightarrow b_{\gamma'}$, but not as an endpoint, then $\gamma$ and $\gamma'$ are intertwined curves.
\end{prop}
\begin{proof}
    Let $d$ be a complete distance on $\hat{\mathcal{M}}$. Since $\psi(\gamma(t)) \rightarrow p$ as $t \rightarrow b_{\gamma}$, there exists an increasing infinite sequence of real numbers $\{x_i\}_{i \in \mathbb{N}}$ in $[0,b_{\gamma})$ with $x_i \rightarrow b_{\gamma}$ as $i \rightarrow \infty$ such that $\psi(\gamma(x_i)) \rightarrow p$ as $i \rightarrow \infty$. Since $p$ is a limit point of the curve $\psi(\gamma')$, there exists an increasing infinite sequence of real numbers $\{y_i\}_{i \in \mathbb{N}}$ in $[0,b_{\gamma'})$ with $y_i \rightarrow b_{\gamma'}$ as $i \rightarrow \infty$ such that $\psi(\gamma'(y_i)) \rightarrow p$ as $i \rightarrow \infty$. Since the two sequences in $\psi(\mathcal{M})$ both end at $p$, 
    \begin{align*}
        d(\psi(\gamma(x_i)),\psi(\gamma'(y_i))) \rightarrow 0~~\text{as}~~i \rightarrow \infty.
    \end{align*}
    Now because $p$ is a limit point of $\psi(\gamma')$, and not an endpoint, there exists an $\epsilon > 0$ and an increasing infinite sequence of real numbers $\{z_i\}_{i \in \mathbb{N}}$ in $[0,b_{\gamma'})$ with $z_i \rightarrow b_{\gamma'}$ as $i \rightarrow \infty$ such that for all $i$, $\psi(\gamma'(z_i)) \notin B_{\epsilon}(p)$. This means that $\epsilon$ is a lower bound of $d(\psi(\gamma(x_i)),\psi(\gamma'(z_i)))$ as $i \rightarrow \infty$. Thus $\gamma$ and $\gamma'$ are intertwined curves.
\end{proof}

\begin{cor} \label{cor:0}
    Let $(\mathcal{M},g,\hat{\mathcal{M}},\psi,\mathcal{C})$ be an envelopment of a pseudo-Riemannian manifold  $(\mathcal{M},g,\mathcal{C})$. Consider the two curves in $\mathcal{C}$, $\gamma:[0,b_{\gamma}) \rightarrow \mathcal{M}$ and $\gamma':[0,b_{\gamma'}) \rightarrow \mathcal{M}$ which are inextendible curves without limit points in $\mathcal{M}$. If both curves $\psi(\gamma(t))$ and $\psi(\gamma'(t))$ approach $p \in  \partial\psi(\mathcal{M})$ as $t \rightarrow b_{\gamma}$ and $t \rightarrow b_{\gamma'}$, but not as endpoints, then $\gamma$ and $\gamma'$ are intertwined curves.
\end{cor}
\begin{proof}
    The proof follows the proof of Proposition \ref{prop:intertwined}. Although the curve $\psi(\gamma(t))$ now approaches $p$, but not as an endpoint, there still exists an increasing infinite sequence of real numbers $\{x_i\}_{i \in \mathbb{N}}$ in $[0,b_{\gamma})$ with $x_i \rightarrow b_{\gamma}$ as $i \rightarrow \infty$ such that $\psi(\gamma(x_i)) \rightarrow p$ as $i \rightarrow \infty$. The rest of the proof is the same.
\end{proof}

Using Definition \ref{defn:Sintertwined} together with Proposition \ref{prop:intertwined} and Corollary \ref{cor:0} we can now provide a complete classification of all pairs of curves which approach a boundary point of an envelopment.

\begin{prop}\label{prop:intertwined_classification}
    Let $(\mathcal{M},g,\hat{\mathcal{M}},\psi,\mathcal{C})$ be an envelopment of a pseudo-Riemannian manifold  $(\mathcal{M},g,\mathcal{C})$. Consider the two curves in $\mathcal{C}$, $\gamma:[0,b_{\gamma}) \rightarrow \mathcal{M}$ and $\gamma':[0,b_{\gamma'}) \rightarrow \mathcal{M}$ which are inextendible curves without limit points in $\mathcal{M}$. If both curves $\psi(\gamma(t))$ and $\psi(\gamma'(t))$ approach $p \in  \partial\psi(\mathcal{M})$ as $t \rightarrow b_{\gamma}$ and $t \rightarrow b_{\gamma’}$, then:

\begin{enumerate}
        \item If both curves have $p$ as an endpoint, $\gamma$ and $\gamma'$ are not intertwined curves,
        \item If one curve approaches $p$ as an endpoint, but the other does not, $\gamma$ and $\gamma'$ are intertwined curves,
        \item If both curves approach $p$, but not as an endpoint, $\gamma$ and $\gamma'$ are intertwined curves.
    \end{enumerate}
\end{prop}

\begin{proof}

\begin{enumerate}
  \item Follows directly from Definition \ref{defn:Sintertwined}
  \item Proposition \ref{prop:intertwined}
  \item Corollary \ref{cor:0}
 \end{enumerate}
\end{proof}


The following proposition has significant utility in our quest to find optimal embeddings for pseudo-Riemannian manifolds. For two inextendible curves without limit points in the manifold $\mathcal{M}$, which intersect at most a finite number of times, it enables the construction of two envelopments where one curve has an endpoint in the boundary of the first envelopment, and the second curve has an endpoint in the boundary of the second envelopment, and the two boundary points are separate (Definition \ref{defn:separated}). This means, for example, that if one curve has bounded parameter, and the other curve has unbounded parameter, with the new envelopments we are effectively separating out a point at infinity from a regular boundary point or a singularity. It also means that if, in a particular envelopment, a pair of curves which approach a boundary point $p$ is intertwined, they can be unravelled with the construction of the two new embeddings.
\begin{prop}\label{prop:simple separate}
    Let $\mathcal{M}$ be an $n$-dimensional, smooth, connected, Hausdorff, paracompact manifold. If $\gamma:[a,b) \rightarrow \mathcal{M}$ and $\gamma':[a',b') \rightarrow \mathcal{M}$ are inextendible curves in $\mathcal{M}$ without limit points, which do not intersect each other infinitely many times, then there exists n-dimensional, smooth, connected, Hausdorff, paracompact manifolds $\mathcal{N}$ and $\mathcal{N}'$ and open embeddings $\psi: \mathcal{M} \rightarrow \mathcal{N}$ and $\psi': \mathcal{M} \rightarrow \mathcal{N'}$ such that $\psi(\gamma) \rightarrow p \in \partial\psi(\mathcal{M})$ and $\psi'(\gamma') \rightarrow q \in \partial\psi'(\mathcal{M})$ and $p$ and $q$ are separate.
\end{prop}
\begin{proof}
    Choose parameter values $c \in [a,b)$ and $c' \in [a',b')$ such that $\gamma|_{[c,b)}$ and $\gamma'|_{[c',b')}$ do not intersect.\\

    After a suitable reparametrisation we have the resulting curves $\gamma:[0,1) \rightarrow \mathcal{M}$ and $\gamma':[0,1) \rightarrow \mathcal{M}$ which do not intersect.\\

    Since $\gamma$ and $\gamma'$ do not intersect, utilising the proof of the Endpoint Theorem \cite{Scott_2021}, smooth scaling functions $f$ and $f'$ can be chosen such that the charts $\mu$ and $\mu'$ have empty intersection.\\

    The new manifolds $\mathcal{N}$ and $\mathcal{N}'$ and open embeddings $\psi: \mathcal{M} \rightarrow \mathcal{N}$ and $\psi': \mathcal{M} \rightarrow \mathcal{N'}$  are constructed such that $\psi(\gamma(t)) \rightarrow p \in \partial\psi(\mathcal{M})$ and $\psi'(\gamma'(t')) \rightarrow q \in \partial\psi'(\mathcal{M})$. Now, continuing to follow the proof of the Endpoint Theorem, for the new charts $\chi$ and $\chi'$, $\mu = \psi^{-1} (\chi \cap \psi(\mathcal{M}))$ and $\mu' = \psi'^{-1}(\chi' \cap \psi'(\mathcal{M}))$ with $\mu \cap \mu' = \emptyset$. Thus $p$ and $q$ are separate.
\end{proof}

    This proof relies only on the part of the proof in \cite{Scott_2021} involving the construction of the normal neighbourhood along a curve.

\subsection{The structure of a pure singularity}
In the study of the abstract boundary construction, the cover relation between abstract boundary points has been well investigated, except for the topological relationship between a pure singularity and a pure point at infinity. In this section we provide insight into the structure of pure singularities of maximally extended pseudo-Riemannian manifolds. In the case where $\mathcal{C}$ is the class of geodesics with affine parameter, and intertwined geodesics are not involved, it will be seen that a complete analysis of this structure is possible.

\begin{defn}
    A {\it pure singularity set} $B$ is a boundary set such that $q$ is a pure singularity for all $q \in B$.
\end{defn}


A $C^l (0 \leq l \leq k)$ pure singularity in an envelopment of a maximally extended $C^k$ pseudo-Riemannian manifold is locally $C^j$ inextendible, where $ j \geq l$. A $C^l$ pure singularity, by definition, will not be covered by any $C^l$ non-singular boundary set containing at least one $C^l$ regular boundary point, nor does it cover any $C^l$ regular boundary point, so it cannot be equivalent to any $C^l$ regular boundary point.

\begin{prop}\label{prop:finite}
    Let $(\mathcal{M},g,\hat{\mathcal{M}},\psi,\mathcal{C})$ be an envelopment of a maximally extended pseudo-Riemannian manifold $(\mathcal{M},g)$. If the boundary point $p \in \partial\psi(\mathcal{M})$ is only approached by curves with bounded parameter, then $p$ is a pure singularity.
\end{prop}
\begin{proof}
     The fact that $(\mathcal{M},g)$ is maximally extended and that the boundary point $p$ is approached by a curve with bounded parameter implies that $p$ is an essential singularity. We need to check if the essential singularity is a directional singularity which covers a pure point at infinity, since from the definition of a directional singularity and Theorem \ref{theorem:1}, for a maximally extended pseudo-Riemannian manifold, $p$ can only cover a pure point at infinity $q$. This cannot be the case, however, since, by definition, $q$ must be approached by a curve with unbounded parameter, and $p$ covers $q$, so it would be the case that the curve also approaches $p$ with unbounded parameter. This would contradict the assumption of the proposition. Thus $p$ is not a directional singularity, and so it is a pure singularity.
\end{proof}
\begin{prop}\label{prop:main}
    Let $(\mathcal{M},g,\hat{\mathcal{M}},\psi,\mathcal{C})$ be an envelopment of a maximally extended pseudo-Riemannian manifold $(\mathcal{M},g)$, where $\mathcal{C}$ is the class of geodesics with affine parameter, with $p \in \partial(\psi(\mathcal{M}))$ a pure singularity, and $\gamma: [0,b) \rightarrow \mathcal{M}$ (where $b \in \mathbb{R}^+ \cup \{+\infty\}$) is a non-self-intersecting geodesic in $\mathcal{C}$ without limit points in $\mathcal{M}$. Suppose $\psi(\gamma) \rightarrow p$. By Proposition \ref{prop:3}, there exists another embedding $\psi'$ such that $\psi'(\gamma) \rightarrow q \in \partial(\psi'(\mathcal{M}))$, $p \vartriangleright q$ and no other geodesic in $\mathcal{C}$ ends at $q$. Suppose that no pair of geodesics in $\mathcal{C}$, each of which approaches $q$, is intertwined in $(\mathcal{M},g,\hat{\mathcal{M}}',\psi',\mathcal{C})$. Then $b \in \mathbb{R}^{+}$ (the geodesic $\gamma$ has bounded affine parameter).
\end{prop}

\begin{proof}
Since no pair of geodesics in $\mathcal{C}$, each of which approaches $q$, is intertwined in $(\mathcal{M},g,\hat{\mathcal{M}}',\psi',\mathcal{C})$, by Proposition \ref{prop:intertwined} no geodesic in $\mathcal{C}$ can approach $q$ not as an endpoint, as it would then be intertwined with $\gamma$. This implies that $\gamma$ is the only geodesic which approaches $q$. It cannot be the case that $\gamma$ has unbounded affine parameter, as then $q$ would be a pure point at infinity, and since $p \vartriangleright q$, this would contradict that $p$ is a pure singularity. Thus $b \in \mathbb{R}^{+}$ (the geodesic $\gamma$ has bounded affine parameter).
\end{proof}

\begin{cor}\label{cor:2}
    Let $(\mathcal{M},g,\hat{\mathcal{M}},\psi,\mathcal{C})$ be an envelopment of a maximally extended pseudo-Riemannian manifold $(\mathcal{M},g)$, where $\mathcal{C}$ is the class of geodesics with affine parameter, with $B \subset \partial(\psi(\mathcal{M}))$ a pure singularity set, and $\gamma: [0,b) \rightarrow \mathcal{M}$ (where $b \in \mathbb{R}^+ \cup \{+\infty\}$) is a non-self-intersecting geodesic in $\mathcal{C}$ without limit points in $\mathcal{M}$. Suppose that there exists a point $p \in B$ such that $\psi(\gamma) \rightarrow p$. By Proposition \ref{prop:3}, there exists another embedding $\psi'$ such that $\psi'(\gamma) \rightarrow q \in \partial(\psi'(\mathcal{M}))$, $p \vartriangleright q$ and no other geodesic in $\mathcal{C}$ ends at $q$. Suppose that no pair of geodesics in $\mathcal{C}$, each of which approaches $q$, is intertwined in $(\mathcal{M},g,\hat{\mathcal{M}}',\psi',\mathcal{C})$. Then $b \in \mathbb{R}^{+}$ (the geodesic $\gamma$ has bounded affine parameter)  and $\gamma$ (which has the endpoint $q$) is the only geodesic which approaches $q$.
\end{cor}
\begin{proof}
    The proof follows directly from Proposition \ref{prop:main} and Proposition \ref{prop:intertwined}.
\end{proof}
The significance of Proposition \ref{prop:main} and Corollary \ref{cor:2} is that for the geodesic $\gamma$, there only needs to exist one new embedding $\psi'$ provided by Proposition \ref{prop:3} which satisfies the not intertwined geodesics condition at $q$ in order to ensure that $\gamma$ has bounded affine parameter. This further guarantees that $q$ is a pure singularity which is only approached by geodesics with bounded parameter -- indeed, it is only approached  by $\gamma$ which has $q$ as its endpoint.

\subsection{An ordering relation and the existence or non-existence of minimal elements}\label{Section Ordering}

In this section we will explain why it is not always necessary to consider a minimal element under the cover relation, since Proposition \ref{prop:main} may provide a natural termination point for the chain.

Let $(\mathcal{M},g,\mathcal{C},\hat{\mathcal{M}},\psi)$ be an envelopment of a maximally extended pseudo-Riemannian manifold $(\mathcal{M},g)$ where $\mathcal{C}$ is the class of geodesics with affine parameter.  For the abstract boundary construction, boundary points are compared under the equivalence relation using the mutual covering relation. To equip the abstract boundary with a partial order, we will use the order on the set ${\cal B}(\cal M)$ of all abstract boundary points given by 
\begin{align*}
    \alpha \succeq \beta \Leftrightarrow \alpha~\text{covers}~\beta .
\end{align*}
This is both transitive and reflexive and is a partial order on the set of all abstract boundary points.\\

We will now show lower bound candidature for a particular totally ordered subset of the abstract boundary. Given an abstract boundary directional singularity $\sigma_0 = [p_0]$ where $p_0 \in \partial\psi(\mathcal{M})$ (an element of the a-boundary which has a directional singularity representative $p_0$), construct a cover relation chain for some index set $A$ by
\begin{align*}
    C = \{\sigma_{\alpha}\}_{\alpha \in A} \subset {\cal B}(\cal M)
\end{align*}
with
\begin{enumerate}
    \item $\sigma_0 \in C$
    \item $C$ is totally ordered by $\succeq$: for every $\alpha$, $\beta$, $\alpha \leq \beta \Rightarrow \sigma_{\alpha} \succeq \sigma_{\beta}$.
\end{enumerate}
So we are considering a totally ordered chain $C \subset {\cal B}(\cal M)$ (totally ordered by $\succeq$) which contains a directional singularity $\sigma_0 \in C$.\\

A longstanding unresolved question related to the abstract boundary classification, is whether or not a directional singularity must always cover a pure singularity of another embedding. We note that since $p_0 \in \partial\psi(\mathcal{M})$ is a directional singularity, it is approached by a geodesic $\gamma$ with bounded affine parameter. We will require that every element $\sigma_{\alpha}$ in the chain $C$ is approached by $\gamma$. From the abstract boundary classification scheme, this constrains every $\sigma_{\alpha}$, for $0 < \alpha$, to be either a directional singularity or a pure singularity.\\

Suppose that in the chain $C$ there exists a first $\alpha > 0$ such that $\sigma_{\alpha}$ is a pure singularity. This means that for all $\beta > \alpha$, $\sigma_{\beta}$ is also a pure singularity. There is therefore no need to consider a minimal element for this chain, since the existence of $\sigma_{\alpha}$ answers the above question for the directional singularity $\sigma_{0}$, and so we can truncate the chain at $\sigma_{\alpha}$.\\

It is certainly still possible, on the other hand, that $C$ is an infinite chain consisting entirely of directional singularities. We are aware that our ability to answer this question for any given directional singularity relates to the existence, or otherwise, of intertwined geodesics in the envelopments associated with each $\sigma_{\alpha}$, $\alpha > 0$. The following corollary to Proposition \ref{prop:main} provides a resolution in the case where Proposition \ref{prop:3} is used to produce an element $\sigma_{\alpha}$ in the chain for which no pair of geodesics approaching $\sigma_{\alpha}$ is intertwined.\\

\begin{cor}\label{cor:directional}
    Let $(\mathcal{M},g,\hat{\mathcal{M}},\psi,\mathcal{C})$ be an envelopment of a maximally extended pseudo-Riemannian manifold $(\mathcal{M},g)$, where $\mathcal{C}$ is the class of geodesics with affine parameter, with $p \in \partial(\psi(\mathcal{M}))$ a directional singularity, and $\gamma: [0,b) \rightarrow \mathcal{M}$ (where $b \in \mathbb{R}^+$) is a non-self-intersecting geodesic in $\mathcal{C}$ without limit points in $\mathcal{M}$. Suppose $\psi(\gamma) \rightarrow p$. By Proposition \ref{prop:3}, there exists another embedding $\psi'$ such that $\psi'(\gamma) \rightarrow q \in \partial(\psi'(\mathcal{M}))$, $p \vartriangleright q$ and no other geodesic in $\mathcal{C}$ ends at $q$. Suppose that no pair of geodesics in $\mathcal{C}$, each of which approaches $q$, is intertwined in $(\mathcal{M},g,\hat{\mathcal{M}}',\psi',\mathcal{C})$. Then $\gamma$ (which has the endpoint $q$) is the only geodesic which approaches $q$ and $q$ is a pure singularity.
\end{cor}

\begin{proof}
The proof proceeds as for the proof of Proposition \ref{prop:main} up to the point where it is established that $\gamma$ is the only geodesic which approaches $q$. It approaches $q$ as an endpoint and, by assumption, has bounded parameter. It cannot be the case that $q$ covers a pure point at infinity of another embedding, as this would imply the existence of a geodesic with unbounded parameter which approaches the pure point at infinity, and therefore also would approach $q$. Since $\gamma$ is the only geodesic which approaches $q$, it follows that $q$ is a pure singularity.
\end{proof}

Note that for this corollary it is assumed that the geodesic $\gamma$ has bounded parameter and that $\psi(\gamma) \rightarrow p$. We recall that Proposition \ref{prop:3} still applies when the geodesic $\gamma$ has unbounded parameter. In this case, if the other conditions of Corollary \ref{cor:directional} are satisfied, then $q$ is a pure point at infinity. If $q$ corresponds to the element $\sigma_{\alpha}$ of the chain $C$, then this is where the chain should be truncated. Although this chain does not provide information regarding the pure singularity part of the directional singularity, it does separate out the pure point at infinity, which is also an important part of unravelling directional singularities and the desired provision of an optimal embedding for the pseudo-Riemannian manifold.


\subsection{The application of Proposition \ref{prop:3} to future g-boundary extensions}\label{section 5}
The partially ordered relation mentioned in \cite{Graf_2024} is equivalent to the cover relation in the abstract boundary framework, so the maximal element of the g-boundary in Graf and Beld-Serrano's paper is equivalent to the maximal element under the cover relation.\\

We proceed now to the definition for the maximal g-boundary defined in the abstract boundary context.\\
\begin{defn}(Future g-boundary Extension \cite{Graf_2024}).
    Let $(\mathcal{M},g)$ be a $C^2$ spacetime. We say that a topological manifold $\mathcal{M}_g$ with boundary is a \textit{future g-boundary extension} of $(\mathcal{M},g)$ if there exists a homeomorphism
    \begin{align*}
        \psi_g: \mathcal{M} \rightarrow int{\mathcal{M}_g}, \text{and}
    \end{align*}
    \begin{enumerate}
        \item for any $p \in \partial\mathcal{M}_g$ there exists a future directed timelike geodesic $\gamma:[0,1) \rightarrow \mathcal{M}$ with $p = \lim_{t \rightarrow 1^-}\psi_g(\gamma(t))$,
        \item all timelike thickenings $O^{\mathcal{M}_g}_{X,r}$ are open and for any $p \in \partial\mathcal{M}_g$ and any future directed timelike geodesic $\gamma:[0,1) \rightarrow \mathcal{M}$ with $p = \lim_{t \rightarrow 1^-}\psi_g(\gamma(t))$ the collection $\{O^{\mathcal{M}_g}_{\dot{\gamma}(1-\frac{1}{n}),\frac{1}{m}}:n,m\in\mathbb{N}\}$ is a neighbourhood basis of $p$.
    \end{enumerate}
\end{defn}
\begin{defn}\label{maximal}
    A future g-boundary extension $\mathcal{M}_g$ of $(\mathcal{M},g)$ is said to be {\it maximal} if any other future g-boundary extension $\hat{\mathcal{M}}$ satisfies $[\mathcal{M}_g] \vartriangleright [\hat{\mathcal{M}}]$.
\end{defn}

    From this definition, any maximal future g-boundary extension automatically has to be unique in the following sense: if $\mathcal{M}_g$ and $\mathcal{M}_g'$ are two maximal future g-boundary extensions, then $[\mathcal{M}_g] = [\mathcal{M}_g']$.\\

We investigate now the boundary point type(s) of a future g-boundary extension using the abstract boundary framework.
\begin{prop}\label{prop:pure g-boundary}
     Let $(\mathcal{M},g,\mathcal{M}_g,\psi_g,\mathcal{C})$ be an envelopment of a $C^{2}$ maximally extended spacetime $(\mathcal{M},g)$, where $\mathcal{M}_g$ is a future g-boundary extension of $(\mathcal{M},g)$ and $\mathcal{C}$ is the family of incomplete future directed timelike geodesics in $(\mathcal{M},g)$. Then the boundary $\partial\psi_g(\mathcal{M})$ of the future g-boundary extension is a pure singularity set.
\end{prop}
\begin{proof}
    All boundary points in $\partial\psi_g(\mathcal{M})$ are the endpoint of an incomplete future directed timelike geodesic, and so by Proposition \ref{prop:finite}, $\partial\psi_g(\mathcal{M})$ is a pure singularity set.
\end{proof}

The following proposition investigates the relationship between future g-boundary extensions and the envelopments produced by Proposition \ref{prop:3}.

\begin{prop}\label{prop:covering}

Let $(\mathcal{M},g,\mathcal{M}_g,\psi_g,\mathcal{C})$ be an envelopment of a $C^{2}$ spacetime $(\mathcal{M},g)$, where $\mathcal{M}_g$ is a future g-boundary extension of $(\mathcal{M},g)$ and $\mathcal{C}$ is the family of incomplete future directed timelike geodesics in $(\mathcal{M},g)$. Let $\gamma: [0,1) \rightarrow \mathcal{M}$ be a geodesic in $\mathcal{C}$ which has the endpoint  $p \in \partial\psi_g(\mathcal{M})$. By Proposition \ref{prop:3}, there exists another envelopment  $(\mathcal{M},g,\mathcal{M}',\psi',\mathcal{C})$ such that $\psi'(\gamma) \rightarrow q \in \partial\psi'(\mathcal{M})$ and $p \vartriangleright q$. Then $p \vartriangleright \partial\psi'(\mathcal{M})$ and $\partial\psi_g(\mathcal{M}) = \mathcal{M}_g\setminus \psi_g(\mathcal{M})$ covers $\partial\psi'(\mathcal{M})$.
\end{prop}
\begin{proof}
    In the proof of Proposition \ref{prop:3}, the smooth scaling function $f$ was chosen such that $f(t) \rightarrow 0$ as $t \rightarrow 1^-$. We use the chart $\mu$ defined in the proof of the Endpoint Theorem \cite{Scott_2021} and employed in the proof of Proposition \ref{prop:3}.\\
    
    Consider a boundary point $r \in \partial\psi'(\mathcal{M})$, where $r$ is not $q$, and let $\{x_i\}$ be a sequence of points in $\mathcal{M}$ such that $\psi'(x_i) \rightarrow r$. There exists an open neighbourhood $U$ of $r$ in $\mathcal{M}'$ such that $$ U \cap \psi'(\mathcal{M}) \subset \psi'(\mu).$$ Since $\psi'(x_i) \rightarrow r$, for some $j$, the sequence elements $\psi'(x_i) \in \psi'(\mu)$ for all $i \geq j$.\\

    This means that $\psi_g(x_i) \in \psi_g(\mu)$ for all $i \geq j$. Since the sequence $\{\psi_g(x_i)\}$ must approach the boundary of $\psi_g(\mathcal{M})$, and for the normal neighbourhood of $\gamma$ used to define the chart $\mu$ the smooth scaling function $f(t) \rightarrow 0$ as $t \rightarrow 1^-$, it must be the case that $\psi_g(x_i) \rightarrow p$. So from Theorem 19 of \cite{scott1994abstract}, $p \vartriangleright r$. Since $p \vartriangleright q$, and $r$ was any other boundary point in $\psi'(\mathcal{M})$, $p \vartriangleright \partial\psi'(\mathcal{M})$. It follows immediately from this that $\partial\psi_g(\mathcal{M}) = \mathcal{M}_g\setminus \psi_g(\mathcal{M})$ covers $\partial\psi'(\mathcal{M})$. 
    
\end{proof}

\section{A Spacetime Without Intertwined Geodesics}\label{section 6}

The purpose of this section is to introduce a physical motivation for how Proposition \ref{prop:main} operates in the Schwarzschild spacetime. We recall the notions of the strongly attached point topology and the black region built from singularly compact future sets. Definition 5.1, Proposition 5.2 and Definition 5.3 come from Barry and Scott \cite{Barry_2014}.
\begin{defn}
    Let $p \in \partial\phi(\mathcal{M})$ be a boundary point of an open embedding $\phi: \mathcal{M} \rightarrow \hat{\mathcal{M}}$. We say that $p$ is \textit{strongly attached} to an open subset $U \subset \mathcal{M}$, if there exists an open neighbourhood $N$ of $p$ in $\hat{\mathcal{M}}$ such that $N \cap \phi(\mathcal{M}) \subseteq \phi(U)$.
\end{defn}

\begin{prop}
    Let $B \subset \partial(\phi(\mathcal{M}))$ be strongly attached to an open set $U \subset \mathcal{M}$, and let $B'$ be a boundary set of a second envelopment $\phi': \mathcal{M} \rightarrow \mathcal{M}'$. If $B \vartriangleright B'$, then $B'$ is also strongly attached to $U$.
\end{prop}

\begin{defn}
    Let $\mathcal{B}(\mathcal{M})$ be the abstract boundary of $\mathcal{M}$. The strongly attached point topology $\mathcal{T}_{sap}(\mathcal{M})$ on $\overline{\mathcal{M}} = \mathcal{M} \cup \mathcal{B}(\mathcal{M})$ is the topology with basis
    \begin{equation}
        \mathcal{W} := \{U \cup \mathcal{B}_{U} | U \subset \mathcal{M}\text{ is open}\},
    \end{equation}
\end{defn}
\noindent where $\mathcal{B}_{U}$ denotes the collection of abstract boundary points which are strongly attached to $U$.\\

Definitions \ref{defn:singular neighourhood}, \ref{defn:singularly compact} and \ref{blackhole from singularity} come from \cite{Wheeler_2023}.

\begin{defn}\label{defn:singular neighourhood}
    An open set $U \subset \mathcal{M}$ is called a singular neighbourhood if every pure singularity abstract boundary point is strongly attached to $U$. Equivalently, $U = \tilde{U} \cap \mathcal{M}$, where $\tilde{U} \subset \Bar{\mathcal{M}}$ is a neighbourhood of $S_p(\mathcal{M})$ (the collection of pure singularity abstract boundary points) in the topology of $\mathcal{T}_{sap}(\mathcal{M})$. 
\end{defn}
\begin{defn}\label{defn:singularly compact}
    Let $(\mathcal{M},g)$ be a Lorentzian manifold. A closed set $A \subset \mathcal{M}$ is called singularly compact if $A \backslash U$ is compact for every singular neighbourhood $U \in \mathcal{U}$, where $\mathcal{U}$ is the family of singular neighbourhoods.
\end{defn}
\begin{defn}\label{blackhole from singularity}
    Let $\mathcal{F}$ be the family of singularly compact future sets, i.e.\ singularly compact sets $A \subset \mathcal{M}$ satisfying $J^+(A) = A$. The black region $\mathscr{B} \subset \mathcal{M}$ is given by
    \begin{align}
        \mathscr{B} \coloneq \bigcup_{A \in \mathscr{F}}A.
    \end{align}
    A connected component $\mathcal{B}$ of $\mathscr{B}$ is a \textit{black hole}. Its boundary $H  \coloneq \partial \mathcal{B}$ is the \textit{event horizon}.
\end{defn}

\subsection{The Schwarzschild spacetime}
We will now investigate how Proposition \ref{prop:main} applies to the maximally extended Schwarzschild spacetime, which will then provide a sufficient condition for the separability between the pure singularities and the pure points at infinity. Consequently, for the Schwarzschild spacetime, when considering the class of all geodesics with affine parameter, one can determine that the Penrose embedding is an optimal embedding (see Barry and Scott \cite{Barry_2014}, \cite{Barrythesis_2014}). This will then permit an examination of how the Penrose embedding relates to Definition \ref{blackhole from singularity}.\\

We establish that the curvature singularities at $r = 0$ in the Schwarzschild spacetime are pure singularities (in the sense of the abstract boundary) and that there are no intertwined geodesics limiting to $r = 0$. \\


Consider the Schwarzschild spacetime $(\mathcal{M},g)$, where $\mathcal{M} = \mathbb{R}^2 \times S^2$ with $(t,r,\theta,\phi)$ coordinates and Lorentzian metric
\begin{align*}
    g = -\vast(1-\frac{2M}{r}\vast)dt^2 + \frac{1}{(1-\frac{2M}{r})}dr^2 + r^2(d\theta^2 + \sin^2{\theta}d\phi^2).
\end{align*}
The curvature invariant $R_{ijkl}R^{ijkl} = C/r^6 \rightarrow \infty$ as $r \rightarrow 0$, so all boundary points at $r = 0$ are curvature singularities.\\

{\it We now show that the curvature singularities at $r=0$ are pure singularities.}\\

Along any geodesic with affine parameter $\tau$, there are two Killing integrals
\begin{align*}
    E = \vast(1-\frac{2M}{r}\vast)\frac{dt}{d\tau};  ~J^2 = r^4\vast(\vast(\frac{d\theta}{d\tau}\vast)^2 + \vast(\sin\theta\frac{d\phi}{d\tau}\vast)^2\vast).
\end{align*}

Let $\kappa \in \{-1,0,1\}$ be the causal character of timelike, null, and spacelike geodesics, respectively.
\begin{align*}
    \kappa = -\vast(1-\frac{2M}{r}\vast)\vast(\frac{dt}{d\tau}\vast)^2 + \frac{1}{(1-\frac{2M}{r})}\vast(\frac{dr}{d\tau}\vast)^2 + r^2\vast(\vast(\frac{d\theta}{d\tau}\vast)^2 + \sin^2{\theta}\vast(\frac{d\phi}{d\tau}\vast)^2\vast)
\end{align*}
Hence, for $J \neq 0$,
\begin{align}
    \vast(\frac{dr}{d\tau}\vast)^2 = E^2 + \vast(1-\frac{2M}{r}\vast)\vast(\kappa - \frac{J^2}{r^2}\vast) \sim \frac{2MJ^2}{r^3} ~~\text{as }r \rightarrow 0\label{drdtau},
\end{align}
\begin{align*}
    \frac{dr}{d\tau} \sim - \sqrt{\frac{2MJ^2}{r^3}}.
\end{align*}
The affine parameter from a sufficiently small $r = r_0$ to $r = 0$ is finite:
\begin{align}
    \tau = \int^{r_0}_{0}\frac{dr}{|dr/d\tau|} \sim \frac{1}{\sqrt{2MJ^2}}\int^{r_0}_{0}r^{\frac{3}{2}}dr = \frac{1}{5}\sqrt{\frac{2}{MJ^2}}r_0^{\frac{5}{2}}.
\end{align}
For radial timelike or null geodesics ($J = 0$) one similarly finds
\begin{align*}
    \frac{dr}{d\tau} \sim \sqrt{\frac{2M}{r}} \Rightarrow \tau &= \frac{1}{3}\sqrt{\frac{2}{M}}r_0^{\frac{3}{2}} ~~(\text{timelike})\\
    \frac{dr}{d\tau} = E \Rightarrow \tau &= \frac{r_0}{E}~~(\text{null}).
\end{align*}
For radial spacelike geodesics, we have
\begin{align*}
    \vast( \frac{dr}{d\tau} \vast) ^2 = E^2 + 1 - \frac{2M}{r} \geq 0 \Rightarrow r \geq r_{min} \equiv \frac{2M}{E^2 + 1} >0.
\end{align*}
Since $r_{min}$ is a positive value, the radial spacelike geodesics do not approach $r=0$. Hence every geodesic terminating at $r=0$ has bounded affine length, and thus, when $\mathcal{C}$ is the class of geodesics with affine parameter, the curvature singularities at $r=0$ are abstract boundary pure singularities by Proposition \ref{prop:finite}.
\\

{\it We next show that the maximally extended Schwarzschild spacetime does not have any intertwined geodesics limiting to $r=0$.}\\

Consider the embedding corresponding to the original Schwarzschild coordinates $(t,r,\theta,\phi)$. The singularities correspond to $(t,0)$ - $\mathbb{R}$. We want to show that there are no intertwined geodesics inside the horizon. All causal geodesics inside the horizon approach $r=0$ with finite affine parameter.
For $J \neq 0$, the geodesic equation near $r=0$ becomes
\begin{align}
     \vast(\frac{dr}{d\tau}\vast)^2 = E^2 + \vast(1-\frac{2M}{r}\vast)\vast(\kappa - \frac{J^2}{r^2}\vast) \sim \frac{2MJ^2}{r^3} ~~\text{as }r \rightarrow 0\label{drdtau1},
\end{align}
where $J^2 = L_x^2 + L_y^2 + L_z^2$ is the total angular momentum and a constant. Moreover, the angular motion follows
\begin{align*}
    \frac{d\Omega}{d\tau} = \frac{J}{r^2}.
\end{align*}

\begin{align}
    \frac{d\Omega}{dr} = \frac{d\Omega/d\tau}{dr/d\tau} \sim (2Mr)^{-\frac{1}{2}}\\
    \Omega(r) \sim \Omega_0 + \vast(\frac{2}{M}r\vast)^{\frac{1}{2}}\label{Angle},
\end{align}
where $d\Omega^2$ is the round metric on $S^2$. The above result implies that for all geodesics (including, obviously, for the radial geodesics), as they approach $r=0$, $\theta \rightarrow$ constant, and $\phi \rightarrow$ constant. That is, all geodesics approach $r=0$ asymptotic to a particular angle. For $J\neq 0$,
\begin{align}
    \frac{dt}{dr} = \frac{dt/d\tau}{dr/d\tau} &\sim  -\frac{rE}{ 2M}\sqrt{\frac{r^3}{2MJ^2}} = -\frac{E}{J}(2M)^{-\frac{3}{2}}r^{\frac{5}{2}},\\
    &t(r) \sim t(r_0) + \frac{2}{7}\frac{E}{J}(2M)^{-\frac{3}{2}}r^{\frac{7}{2}}\label{dt/dr}.
\end{align}
From the integration we see that $t$ limits to a finite value as $r \rightarrow 0$. This is also true for the causal radial geodesics which approach $r=0$. \\

Thus all geodesics inside the horizon which end at $r = 0$ approach some $(t,0)$ asymptotic to a specific angle. For two such geodesics $\gamma$ and $\gamma'$, if $\gamma \rightarrow (t_1,0)$, $\gamma' \rightarrow (t_2,0)$ and $t_1 \neq t_2$, then condition 1 of Definition \ref{defn:Sintertwined} is not satisfied. If $t_1 = t_2$, then condition 2 is not satisfied. This means that no pair of geodesics which approach $r = 0$ is intertwined.

\begin{figure}[htbp]
\centerline{\includegraphics[scale=0.6]{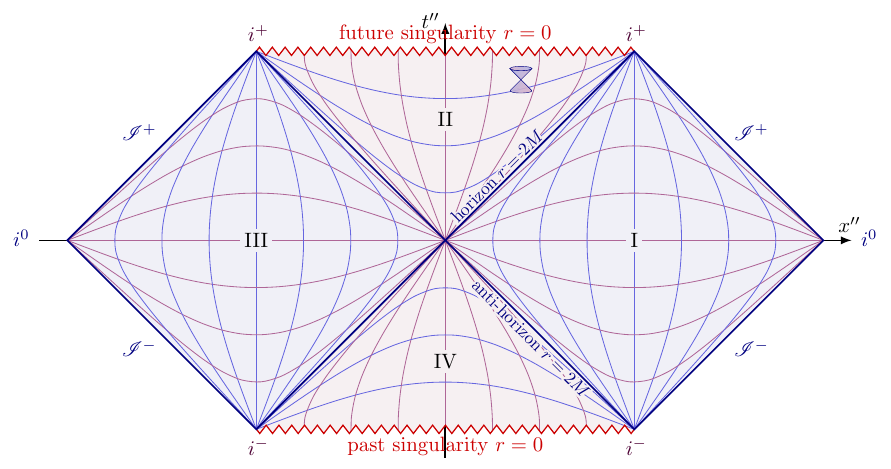}}
\caption{The Penrose diagram for the Kruskal-Szekeres maximal extension of the Schwarzschild spacetime. This diagram depicts the $t = constant$ and $r = constant$ surfaces, and was produced using reference \cite{diagram}.}
\label{fig2}
\end{figure}

The Penrose diagram of the Kruskal-Szekeres maximal extension of the Schwarzschild spacetime utilises compactified advanced and retarded null coordinates $(v'',w'')$ in the region $r \in (0,2M)$:
\begin{align*}
    v'' = \arctan{\Vast(\sqrt{1-\frac{r}{2M}}e^{\frac{r+t}{4M}}\Vast)} \in \vast(0,\frac{\pi}{2}\vast),\\
    w'' = \arctan{\Vast(\sqrt{1-\frac{r}{2M}}e^{\frac{r-t}{4M}}\Vast)} \in \vast(0,\frac{\pi}{2}\vast).
\end{align*}
The new rectangular coordinates $(t'',x'')$ are given by

\begin{align*}
    t'' = v'' + w'' = \arctan{\Vast(\sqrt{1-\frac{r}{2M}}e^{\frac{r+t}{4M}}\Vast)} + \arctan{\Vast(\sqrt{1-\frac{r}{2M}}e^{\frac{r-t}{4M}}\Vast)},\\
    x'' = v'' - w'' = \arctan{\Vast(\sqrt{1-\frac{r}{2M}}e^{\frac{r+t}{4M}}\Vast)} - \arctan{\Vast(\sqrt{1-\frac{r}{2M}}e^{\frac{r-t}{4M}}\Vast)}.
\end{align*}
All geodesics which approach $r=0$ approach a particular point $(\pi/2,x'')$ in the Penrose diagram, because $t'' \rightarrow \pi/2$ as $r \rightarrow 0$, and they do so at a specific angle $\theta$, $\psi$ constant and with finite affine parameter.\\

Every point $(\pi/2, x'')$ in Figure \ref{fig2} is a curvature singularity and also a pure singularity. For two geodesics $\gamma$ and $\gamma'$, if $\gamma \rightarrow (\pi/2,x_1)$, $\gamma' \rightarrow (\pi/2,x_2)$ and $x_1 \neq x_2$, then condition 1 of Definition \ref{defn:Sintertwined} is not satisfied. If $x_1 = x_2$, then condition 2 is not satisfied. This means that no pair of geodesics which approach $t'' = \pi/2$ is intertwined.\\

In conclusion, for the embedding corresponding to the Penrose diagram for the Kruskal-Szekeres maximal extension of the Schwarzschild spacetime, all boundary points in the interval $(\pi/2,x'')$ are curvature singularities and pure singularities, and no pair of geodesics which approach $t''=\pi/2$ is intertwined.\\

The null geodesic at $r=2M$ has an infinite affine parameter, so its endpoint at $(\pi/2,\pi/2)$ (the ``corner'' boundary point $i^{+}$) cannot be a pure singularity, which then clearly delineates the region of the pure singularities of the Schwarzschild spacetime. The following argument provides an explanation.\\

In the region with $r \in (0,2M)$, using the coordinates 
\begin{align*}
t' = \sqrt{2M-r}\exp{(r/4M)}\sinh{(t/4M)}, x' = \sqrt{2M-r}\exp{(r/4M)}\cosh{(t/4M)},
\end{align*}
    the metric takes the form
\begin{align}
    ds^2 = F^2(r)(-dt'^2 + dx'^2) + r^2(d\theta^2 + \sin^2\theta d\phi^2)\label{metric:t'x'}
\end{align}
where $F^2(r) = 16M^2/r\exp(-r/2M)$, which is a regular metric for all $r > 0$.\\

At a fixed angle in the above coordinates, the null condition gives:
    \begin{align}\label{null 2M}
        0 = F^2(r)(-\dot{t}'^2 + \dot{x}'^2) \Rightarrow \frac{dx'}{dt'} = \pm 1.
    \end{align}
This implies that the radial null geodesics lie at $\pm 45$° and, in particular, null geodesics lie along the event horizon at $r = 2M$ and proceed out to the corner point $(\pi/2,\pi/2)$.\\
    
    We now solve the geodesic equation for the metric \ref{metric:t'x'} along $t'=x'$ $(r = 2M)$. For the geodesic equation for $t'$, all associated partial derivatives of the metric vanish there, and so the geodesic equation reduces to
    \begin{align*}
        \frac{d^2t'}{d\tau^2} = 0,
    \end{align*}
    so that
    \begin{align*}
        t'(\tau) = \alpha \tau + \beta .
    \end{align*}
    Since $t'$ ranges over $(-\infty,\infty)$ at $r = 2M$ in the extension, the null geodesics are complete along $r = 2M$.\\
    

\begin{ex}(Spacelike geodesics in the Schwarzschild spacetime) Inside the horizon, we set $r^* = 2M -r$. We will investigate the spacelike geodesic behaviour as $r \rightarrow 2M$, $r^* \rightarrow 0$. For radial spacelike geodesics, we have
    \begin{align}
        1 = -\vast(1-\frac{2M}{r}\vast) \vast(\frac{dt}{d \tau}\vast)^2 + \frac{1}{(1-\frac{2M}{r})}\vast(\frac{dr}{d \tau}\vast)^2,\nonumber\\
        1 = \frac{r^*}{2M - r^*}\frac{E^2(2M - r^*)^2}{(r^*)^2} - \frac{2M - r^*}{r^*}\vast(\frac{dr^*}{d \tau}\vast)^2,\nonumber\\\
        \vast(\frac{dr^*}{d \tau}\vast)^2 = E^2 - \frac{r^*}{2M - r^*} \sim E^2~~\text{as}~r^* \rightarrow 0.
    \end{align}
    We assume $E>0$, then as $r^* \rightarrow 0$, we obtain
    \begin{align*}
        \tau \sim -\frac{1}{E}r^* + \alpha ~~\text{and}~~ \tau \rightarrow \alpha . 
    \end{align*}
    Thus the affine parameter will be finite for all radial spacelike geodesics which approach $r = 2M$ with $r < 2M$.\\
    
    \noindent Let $\tau^* = \tau -\alpha \rightarrow 0^-$ as $r^* \rightarrow 0$.
    \begin{align*}
        \frac{dt}{d \tau^*} \sim -\frac{2M + E\tau^*}{\tau^*}
    \end{align*}
    \begin{align}
        t \sim -2M\ln{(-\tau^*)} - E\tau^* + \beta
    \end{align}
    Thus we have $t \rightarrow \infty$ as $r^* \rightarrow 0$ for all radial spacelike geodesics. The above calculations provide the same results for the non-radial spacelike geodesics.\\

    For non-radial spacelike geodesics, we have
    \begin{align*}
        \vast(\frac{dr^*}{d \tau}\vast)^2 = E^2 - \frac{r^*}{2M - r^*}\vast(1 - \frac{J^2}{(2M - r^*)^2}\vast) \sim E^2~~\text{as}~r^* \rightarrow 0.
    \end{align*}
    For $E>0$, we solve
    \begin{align}\label{r expression}
        \tau \sim -\frac{1}{E}r^* + \alpha. 
    \end{align}
    \noindent Let $\tau^* = \tau -\alpha \rightarrow 0^-$ as $r^* \rightarrow 0$.
    \begin{align}
        \frac{dt}{d \tau^*} \sim -\frac{2M + E\tau^*}{\tau^*}
    \end{align}
    \begin{align}\label{t expression}
        t \sim -2M\ln{(-\tau^*)} - E\tau^* + \beta
    \end{align}
    Thus, we have that $t \rightarrow \infty$ and the affine parameter is finite as $r^* \rightarrow 0$ for all spacelike geodesics.\\

    By Eqn \ref{r expression} and Eqn \ref{t expression}, as $\tau^* \rightarrow 0$, we have
    \begin{align*}
        &r + t \sim 2M - 2M\ln{(-\tau^*)} + \beta,\\
        &r - t \sim 2M + 2M\ln{(-\tau^*)} + 2E\tau^* - \beta.
    \end{align*}
    \begin{align*}
        \sqrt{1 - \frac{r}{2M}}e^{\frac{r+t}{4M}} \sim \sqrt{\frac{E}{2M}}e^{\frac{1}{2}(1 + \frac{\beta}{2M})} \rightarrow C > 0~~\text{as}~\tau^* \rightarrow 0,\\
        \sqrt{1 - \frac{r}{2M}}e^{\frac{r-t}{4M}} \sim \sqrt{\frac{E}{2M}}e^{\frac{1}{2}(1 - \frac{\beta}{2M} + \frac{E}{M}\tau^*)}(-\tau^*) \rightarrow 0^+~~\text{as}~\tau^* \rightarrow 0.
    \end{align*}
    \begin{align}
        (t'',x'') \rightarrow (e,e)~~\text{as}~\tau^* \rightarrow 0.
    \end{align}
    Thus every spacelike geodesic crosses the horizon at $(t'',x'') = (e,e)$ where $ 0 < e < \pi/2$, which implies that $i^+$, $i^-$ are points at infinity in the Penrose maximal extension.
\end{ex}

As depicted in Figure \ref{fig2}, the above calculations show that the future and past singularities are pure singularities, and $\mathscr{I}^{+} \cup \mathscr{I}^{-} \cup \{i^0,i^{+},i^{-}\}$ are pure points at infinity for the class of all geodesics with affine parameter in the Penrose maximal extension of the Schwarzschild spacetime. It implies that the pure singularities are separated from the pure points at infinity, which is an essential element required to show that the maximally extended Schwarzschild spacetime is an optimal embedding.\\


    Consider the Penrose compactification of the Kruskal-Szekeres extension, i.e.\ an embedding $\psi:\mathcal{M} \rightarrow \mathcal{M}'$ whose boundary consists of the null infinities $\mathscr{I}^{\pm}$ together with the points $i^0$,$i^{\pm}$, and the curvature singularities.\\ 
    
    (1) Pure singularities at $r=0$. From the geodesic asymptotic analysis established earlier in this section, every geodesic approaching $r=0$ has a bounded affine parameter. Consequently, the boundary points at $r = 0$ are pure singularities.\\

    (2) Pure points at infinity. We have shown that radial null geodesics along $r=2M$ have an unbounded affine parameter. Since all geodesics which approach $i^{+}$ from $r > 2M$ also have unbounded affine parameter, the "corner" point $i^{+}$ at $(\pi/2,\pi/2)$ is a pure point at infinity. The same conclusion applies for all four “corners” of the Penrose diagram. Thus the subset $\mathscr{I}^{+} \cup \mathscr{I}^{-} \cup \{i^0,i^{+},i^{-}\}$ of $\partial\psi(\mathcal{M})$ is comprised of pure points at infinity.\\


    

We now consider whether the embedding $\psi:\mathcal{M} \rightarrow \mathcal{M}'$ is an {\it optimal embedding} of the Schwarzschild spacetime. We will show that the Penrose compactification of the maximal extension of the Schwarzschild spacetime is an optimal embedding following the definition provided in \cite{Barrythesis_2014}.\\

(1) Partial cross section: since $\psi$ is an envelopment, $\sigma_{\psi} = \{[p]: p \in \partial\psi(\mathcal{M})\}$ is a partial cross section (Defn.\ 7.2.6 of \cite{Barrythesis_2014}). \\

(2) Ideal admissible partial cross section: the Kruskal-Szekeres maximal extension of the Schwarzschild spacetime is at least $C^2$ inextendible. Therefore the set $\mathcal{B}_r(\mathcal{M})$ of all ideal regular abstract boundary points of $(\mathcal{M},g)$ is the empty set. It is evident that for each $[p] \in \sigma_{\psi}$, where $p \in \partial\psi(\mathcal{M})$, there exists an open neighbourhood $\mathcal{N}$ of $p$ such that $\{[s]:s \in \mathcal{N} \cap \partial\psi(\mathcal{M})\} \subset \sigma_{\psi}$. From Defn.\ 7.2.30 of \cite{Barrythesis_2014}, this implies that $\sigma_{\psi}$ is an ideal admissible partial cross section.\\

(3) Maximal ideal admissible partial cross section: suppose that $\sigma_{\psi}$ is not a maximal ideal admissible partial cross section (see Defn.\ 7.2.31 of \cite{Barrythesis_2014}). Then there exists an ideal admissible partial cross section $\sigma$ which contains $\sigma_{\psi}$ as a proper subset, which means that there exists an abstract boundary point $[p]$ which is an element of $\sigma$ but which is not contained in $\sigma_{\psi}$. There exists an envelopment $\phi:\mathcal{M} \rightarrow \mathcal{\hat{M}}$ where $p \in \partial\phi(\mathcal{M})$. Consider a sequence $\{x_i\} \subset \mathcal{M}$ without limit points in $\mathcal{M}$ such that $\phi(x_i) \rightarrow p \in \partial\phi(\mathcal{M})$. Because the Penrose maximal extension of the Schwarzschild spacetime is compact, every infinite sequence in $\psi(\mathcal{M})$ has a limit point in $\overline{\psi(\mathcal{M})}$. Since the sequence $\{x_i\} \subset \mathcal{M}$ has no limit points in $\mathcal{M}$, it follows that there exists an infinite subsequence $\{{x_i}_{j}\} \subset \{x_i\}$ such that $\psi({x_i}_{j}) \rightarrow q \in \partial\psi(\mathcal{M})$. Thus $\phi({x_i}_{j}) \rightarrow p \in \partial\phi(\mathcal{M})$ and $\psi({x_i}_{j}) \rightarrow q \in \partial\psi(\mathcal{M})$ implying that the abstract boundary points $[p]$ and $[q]$ are not separate (Defn. 7.2.3 of \cite{Barrythesis_2014}). Since $\sigma$ is a partial cross section, it must be the case then that $[p] = [q]$ which contradicts that $[p]$ is not contained in $\sigma_{\psi}$. Thus $\sigma_{\psi}$ is a maximal ideal admissible partial cross section.\\

From Defn.\ 7.3.1 of \cite{Barrythesis_2014}, since $\sigma_{\psi}$ is a maximal ideal admissible partial cross section, the envelopment $\psi$ for the Penrose compactification of the maximal extension of the Schwarzschild spacetime is indeed an {\it optimal embedding}.\\

It remains to consider how Wheeler's Definitions \ref{defn:singular neighourhood}, \ref{defn:singularly compact} and \ref{blackhole from singularity} apply to the envelopment $\psi$. Due to the possible presence of directional singularities occurring in other envelopments of the maximally extended Schwarzschild spacetime, we will only consider these definitions with respect to pure singularity abstract boundary points arising from the envelopment $\psi$. Furthermore, we need only consider the pure singularities of this envelopment which lie to the future, i.e.\ with $t''=\pi/2$.\\

In Figure \ref{fig:optimal 1U}, $\psi(U)$ is depicted where $U$ is an open set in $\mathcal{M}$ such that every pure singularity abstract boundary point arising from the envelopment $\psi$ (with $t''=\pi/2$) is strongly attached to $U$. Consider a point $p \in A$ where $A$ is a closed set in $\mathcal{M}$. If $J^+(A) = A$, then $J^+(p) \subset A$. For the envelopment $\psi$, consider the cases when $\psi(p) \in$ region I, III or IV (in Figure \ref{fig:optimal 1U}, $\psi(p) \in$ region IV). In all three cases, $\psi(J^+(p))$ is a closed future null cone based at $\psi(p)$, for which a segment of future null infinity $\mathscr{I}^+$ lies on its future boundary. This means that a segment of future null infinity $\mathscr{I}^+$ lies on the future boundary of $\psi(A)$, and although $A \backslash U$ is closed, it is certainly not bounded in $\mathcal{M}$ and is therefore not compact. So a singulary compact future set $A$ cannot contain points $p \in \mathcal{M}$ where $\psi(p) \in$ regions I, III or IV.

\begin{figure}[H]
    \centering
    \includegraphics[width=0.5\linewidth]{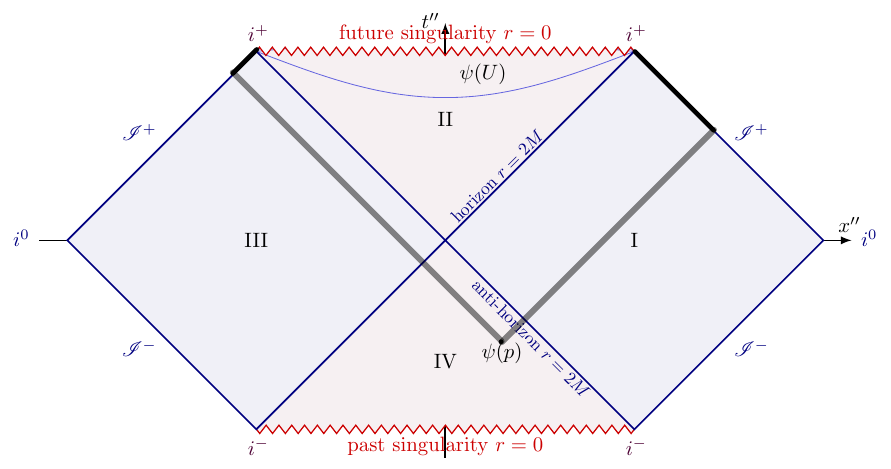}
    \caption{In the Penrose diagram for the Schwarzschild spacetime, for a point $p \in \mathcal{M}$ such that $\psi(p) \in$ region IV, the set $\psi(J^+(p))$ is depicted. It can be seen that two segments of future null infinity $\mathscr{I}^+$ lie on its future boundary. For an open set $U$ in $\mathcal{M}$, the open set $\psi(U)$ is shown, such that every pure singularity abstract boundary point arising from the envelopment $\psi$ (with $t''=\pi/2$) is strongly attached to $U$. Although $J^+(p) \backslash U$ is closed, it is not bounded in $\mathcal{M}$. So $J^+(p)$ is not a singularly compact future set.}
    \label{fig:optimal 1U}
\end{figure}

Now we consider points $p \in \mathcal{M}$ for which $\psi(p) \in$ region II, and closed sets $A=J^+(p)$. For some open sets $U$ which have every pure singularity at $t''=\pi/2$ strongly attached to them, $\psi(U)$ will contain $\psi(A)$, and the others will cut across $\psi(A)$, so that $A \backslash U$ is closed and bounded in $\mathcal{M}$, and therefore compact (see Figure \ref{fig:optimal 2U}). The collection $\mathscr{F}$ of sets $A$ is then a family of singularly compact future sets. Thus the black region $\mathscr{B} \subset \mathcal{M}$ given by
    \begin{align}
        \mathscr{B} \coloneq \bigcup_{A \in \mathscr{F}}A
    \end{align}
is region II which is a connected component $\mathcal{B}$ of $\mathscr{B}$ and is therefore a \textit{black hole} (Definition \ref{blackhole from singularity}). Its boundary $H  \coloneq \partial \mathcal{B}$ is indeed the \textit{event horizon} at $r=2M$.

\begin{figure}[H]
    \centering
    \includegraphics[width=0.5\linewidth]{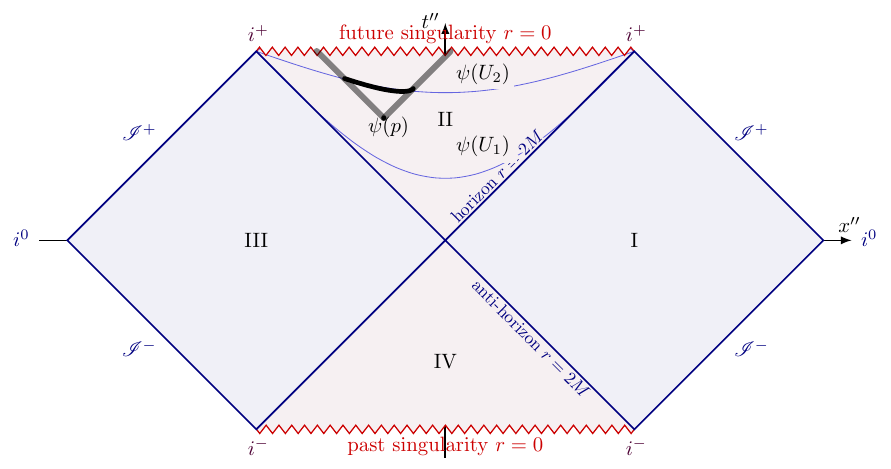}
    \caption{In the Penrose diagram for the Schwarzschild spacetime, for a point $p \in \mathcal{M}$ such that $\psi(p) \in$ region II, the set $\psi(J^+(p))$ is depicted. For open sets $U_1$ and $U_2$ in $\mathcal{M}$, the open sets $\psi(U_1)$ and $\psi(U_2)$ are shown, such that every pure singularity abstract boundary point arising from the envelopment $\psi$ (with $t''=\pi/2$) is strongly attached to $U_1$ and $U_2$. It can be seen that $\psi(J^+(p)) \subset \psi(U_1)$ and so $J^+(p) \backslash U_1 = \emptyset$. The open set $\psi(U_2)$ cuts across $\psi(J^+(p))$ so that $\psi(J^+(p)) \backslash \psi(U_2)$ is the ``triangular'' region depicted, and $J^+(p) \backslash U_2$ is a compact set. The set $J^+(p)$ is a singularly compact future set.}
    \label{fig:optimal 2U}
\end{figure}

\section{Discussion}

The title of this paper ``The relationship between spacetime singularities and regions at infinity" is intimately connected with the following conjecture proposed by Wheeler in \cite{Wheeler_2023}.

\begin{conj}\label{conj:Wheeler}
    Under physically relevant choices of curve families $\mathcal{C}$ on a smooth Lorentzian manifold $(\mathcal{M},g)$, (e.g. geodesics with affine parameter, $C^1$ curves with generalised affine parameter, the causal subfamilies of either of these, etc.), points in $\mathcal{I}(\mathcal{M})$ and $\mathcal{S}_p(\mathcal{M})$ are separated from each other.
\end{conj}
\noindent In this conjecture $\mathcal{I}(\mathcal{M})$ is the collection of abstract boundary pure points at infinity and $\mathcal{S}_p(\mathcal{M})$ is the collection of pure singularity abstract boundary points.\\

For the class $\mathcal{C}$ of geodesics with affine parameter, for example, this conjecture raises the immediate question of what sort of geodesics can approach a pure singularity boundary point. From the definition, a pure singularity must be approached by a geodesic with bounded parameter, but can it also be approached by a geodesic with unbounded affine parameter? If this is the case, the pure singularity boundary point could be in contact with a pure point at infinity of another envelopment, whilst not covering it (which would make it a directional singularity), and so the two associated abstract boundary points would not be separate. The analogous question for pure points at infinity has a straightforward answer; pure points at infinity are only approached by geodesics with unbounded affine parameter.\\

In this paper we have progressed the understanding of these questions by focussing on the class of affinely parametrised geodesics in maximally extended pseudo-Riemannian manifolds. From Theorem \ref{theorem:1} we know that these manifolds may have an associated abstract boundary point which is a directional singularity and, if this is the case, it must cover an abstract boundary point which is a pure point infinity. By definition, a directional singularity must also be approached by a geodesic with bounded parameter, and it remains an open question as to whether, in general, the directional singularity must always cover a pure singularity of another envelopment. If it does cover a pure singularity, are the pure point at infinity and the pure singularity covered by the directional singularity necessarily separate? \\

In Section \ref{section 3} Proposition \ref{prop:3}, which uses the Endpoint Theorem, establishes a very general result for an $n$-dimensional, pseudo-Riemannian manifold $(\mathcal{M},g)$. If there exists a geodesic $\gamma$ which, under an envelopment of $(\mathcal{M},g)$, ends at a boundary point $p$, then there always exists another envelopment in which the geodesic ends at the boundary point $q$, no other geodesic ends at $q$, and the boundary point $p$ covers the boundary point $q$. Obviously this result has great utility in separating out different components of boundary points, for example, when a geodesic with bounded parameter ends at a directional singularity.\\

The concept of a pair of intertwined geodesics is considered in Section \ref{section 4}. First introduced by Chru\'{s}ciel, the definition proposed by Graf and Beld-Serrano is given, which is envelopment independent, followed by a new definition that we provide based on the abstract boundary framework, which relates to a particular envelopment, and is therefore envelopment dependent. It transpires that the non-existence of intertwined geodesics is critical to our ability to answer the questions posed in this discussion.\\

Using the new definition, for maximally extended pseudo-Riemannian manifolds, when the situation detailed two paragraphs above is considered, for the case when $p$ is a pure singularity, Corollary \ref{cor:2} shows that if no pair of geodesics which approaches $q$ is intertwined, then the geodesic $\gamma$ has bounded parameter, and no other geodesic approaches $q$. This means that, from the pure singularity $p$, we have separated out the pure singularity $q$ which is only approached by geodesics with bounded parameter (indeed it is only approached by $\gamma$ which has the endpoint $q$). Furthermore, in Corollary \ref{cor:directional} we establish that if $p$ is a directional singularity and the geodesic $\gamma$ which ends at $p$ has bounded parameter, then $q$ is a pure singularity and $\gamma$ is the only geodesic which approaches $q$. So in this case, from a directional singularity, we have produced a pure singularity $q$ which is only approached by one geodesic with bounded parameter.\\

In Section \ref{Section Ordering}, when we considered a chain of abstract boundary points ordered by the covering relation, we examined the case where one of the elements of the chain is a directional singularity abstract boundary point $[p]$, and $p$ is the endpoint of a geodesic $\gamma$ with bounded parameter. Then if Proposition \ref{prop:3} can be employed to form another envelopment, where $\gamma$ ends at the boundary point $q$, which forms the next element [q] of the chain, and no pair of intertwined geodesics approaches $q$, then $[q]$ is a pure singularity abstract boundary point which is only approached by the geodesic $\gamma$. Also, if the geodesic $\gamma$ has unbounded parameter, then $[q]$ is a pure point at infinity abstract boundary point. In both cases, it is the production of an envelopment without intertwined geodesics which enables us to truncate the covering chain at either a pure singularity or pure point at infinity abstract boundary point. In Section \ref{section 5} we present a definition for the maximal g-boundary in the abstract boundary context, and show that for a maximally extended spacetime $(\mathcal{M},g)$, the boundary of the future g-boundary extension is a pure singularity set.\\

In Section \ref{section 6} we analyse the definitions and results of the preceding sections with respect to the Schwarzschild solution. In particular, we consider the envelopment corresponding to the Penrose diagram for the maximally extended Schwarzschild spacetime. We establish that there are no intertwined geodesic pairs inside the event horizon, and thereby determine for the Penrose diagram, that all boundary points $(\pi/2,x'')$ and $(-\pi/2,x'')$, where $-\pi/2<x''<\pi/2$, are curvature singularities and abstract boundary pure singularities which are only approached by geodesics with bounded parameter. The remaining boundary points are all pure points at infinity. This enables us to show that the Penrose diagram is an optimal embedding in the abstract boundary context.\\

Finally, we test Wheeler's definitions of a black hole and an event horizon for the Penrose diagram optimal embedding and find that they work well and produce the expected result. We did, however, use a slightly modified definition of a singular neighbourhood $U$, requiring only that every pure singularity abstract boundary point arising from the Penrose diagram is strongly attached to $U$. It is not possible to have complete knowledge of all envelopments of a spacetime, including if a directional singularity exists in one of them, and whether it covers a pure singularity in another envelopment.\\

We believe that Wheeler's definitions will be amenable to application and work well when they are restricted to an optimal embedding for a solution. This brings us back to Conjecture \ref{conj:Wheeler} which we don't think will be true in general, as there may exist envelopments with a boundary point which is a pure singularity approached by a geodesic with unbounded affine parameter, raising the possibility that it is in contact with a pure point at infinity of another envelopment (and so not separate). In a future paper we will explore this possibility further, and will look for conditions other than the non-existence of pairs of intertwined geodesics to separate out a pure singularity from a directional singularity, and to obtain pure singularities which are only approached by geodesics with bounded parameter.\\


\bmhead{Acknowledgements}
Junbang Liu acknowledges the support of the ANU HDR Fee Merit Scholarship and the University Research Scholarship. The research of Susan M.\ Scott is supported by the Australian Research Council Centre of Excellence for Gravitational Wave Discovery (OzGrav), project number CE230100016. We also thank Joan Licata and Xintao Luo for fruitful discussions.

\section*{Declarations} {\bf Statements and Declarations: Competing Interests and Data Availability} On behalf of all authors, the corresponding author states that there is no conflict of interest and data sharing is not applicable to this article as no datasets were generated or analysed during the current study.


\newpage
\bibliography{sn-bibliography.bib}

\end{document}


\begin{tikzpicture}[>=Stealth,scale=1.2]

\begin{scope}[shift={(4,0)}]

\draw[thick] (-1,-3) -- (-1,3);
\draw[thick] (1,-3) -- (1,3);

\draw[dashed] (0,2) ellipse (1 and 0.35);
\draw[dashed] (0,0) ellipse (1 and 0.35);
\draw[dashed] (0,-2) ellipse (1 and 0.35);



\draw[orange,thick,->,samples=600,variable=\t,domain=0:4]
 plot ({ (1 + 0.0*exp(-1.0*\t))*cos(10*pi*\t r) },
       { -2*exp(-0.9*\t) + 0.35*(1 + 0.6*exp(-1.0*\t))*sin(10*pi*\t r) } );

\draw[blue,thick,->,samples=600,variable=\t,domain=0:4]
 plot ({ (1 + 0.0*exp(-1.0*\t))*cos(10*pi*\t r) },
       { -2*exp(-0.9*\t) + 0.35*(1 + 0.3*exp(-1.0*\t))*sin(10*pi*\t r) } );

\draw[black,very thick] (0,0) ellipse (1 and 0.35);

\node[right] at (1.2,0) {$t=0$};

\node at (0.0,-3.3) {$\psi'$};
\node at (0.0,3.2) {$(c)$};

\node at (-0.0,-2.1) {$\gamma_2$};
\node at (0.0,-2.5) {$\gamma_1$};

\end{scope}

\begin{scope}[shift={(0,0)}]

\draw[thick] (-1,-3) -- (-1,3);
\draw[thick] (1,-3) -- (1,3);

\draw[dashed] (0,2) ellipse (1 and 0.35);
\draw[dashed] (0,0) ellipse (1 and 0.35);
\draw[dashed] (0,-2) ellipse (1 and 0.35);


\draw[green,thick,->,samples=600,variable=\t,domain=0:4]
 plot ({ (1 + 0.0*exp(-1.0*\t))*cos(10*pi*\t r) },
       { -2*exp(-0.9*\t) + 0.35*(1 + 0.6*exp(-1.0*\t))*sin(10*pi*\t r) } );

\draw[black,very thick] (0,0) ellipse (1 and 0.35);

\node[right] at (1.2,0) {$t=0$};

\node at (0.0,-3.3) {$\psi$};
\node at (0.0,3.2) {$(b)$};
\node[right] at (0.35,-2.02) {$\gamma_3$};
\node[right] at (-0.3,1.3) {$\gamma_1$};

\end{scope}

\draw[thick,->] (6,-3) -- (6,3) node[above] {$t$};

\begin{scope}[shift={(-4,0)}]

\draw[thick] (-1,-3) -- (-1,3);
\draw[thick] (1,-3) -- (1,3);

\draw[dashed] (0,2) ellipse (1 and 0.35);
\draw[dashed] (0,0) ellipse (1 and 0.35);
\draw[dashed] (0,-2) ellipse (1 and 0.35);



\draw[white,very thick] (0,0) ellipse (1 and 0.35);

\draw[thick] (0,0) ellipse (1 and 0.35);
\node at (0.0,-3.3) {$\psi$};
\node at (0.0,3.2) {$(a)$};

\draw[-,line width=1.0pt,blue] (0.3 cm, -3 + 0.75) -- ++(0, 6 - 1.2) ;
\draw[-,line width=1.0pt,orange] (-0.3 cm, -3 + 0.75) -- ++(0, 6 - 1.2) ;
\draw[-,line width=1.0pt,orange] (3.7 cm, -3 + 0.75) -- ++(0, 6 - 1.2);

  

\node[left] at (-0.32,-0.7) {$\gamma_1$};
\node[right] at (0.32,-0.7) {$\gamma_2$};

\end{scope}

\end{tikzpicture}


\begin{tikzpicture}[scale=3.2]
  \message{Extended Penrose diagram: Schwarzschild black hole^^J}
  
  \def\R{0.08} 
  \def\Nlines{3} 
  \pgfmathsetmacro\ta{1/sin(90*1/(\Nlines+1))} 
  \pgfmathsetmacro\tb{sin(90*2/(\Nlines+1))}   
  \pgfmathsetmacro\tc{1/sin(90*2/(\Nlines+1))} 
  \pgfmathsetmacro\td{sin(90*1/(\Nlines+1))}   
  \coordinate (-O) at (-1, 0); 
  \coordinate (-S) at (-1,-1); 
  \coordinate (-N) at (-1, 1); 
  \coordinate (-W) at (-2, 0); 
  \coordinate (-E) at ( 0, 0); 
  \coordinate (O)  at ( 1, 0); 
  \coordinate (S)  at ( 1,-1); 
  \coordinate (N)  at ( 1, 1); 
  \coordinate (E)  at ( 2, 0); 
  \coordinate (W)  at ( 0, 0); 
  \coordinate (B)  at ( 0,-1); 
  \coordinate (T)  at ( 0, 1); 
  \coordinate (X0) at ({asin(sqrt((\ta^2-1)/(\ta^2-\tb^2)))/90},
                       {-acos(\ta*sqrt((1-\tb^2)/(\ta^2-\tb^2)))/90}); 
  \coordinate (X1) at ({asin(sqrt((\tc^2-1)/(\tc^2-\td^2)))/90},
                       {acos(\tc*sqrt((1-\td^2)/(\tc^2-\td^2)))/90}); 
  \coordinate (X2) at (45:0.87); 
  \coordinate (X3) at (0.60,1.05); 
  \coordinate (X4) at (0.4,0.85);

  \draw[->,thick] (0,-1.15) -- (0,1.15) node[above=1,left=-1] {$t''$};
  \draw[->,thick] (-2.15,0) -- (2.15,0) node[left=1,above=0] {$x''$};
  
  \begin{scope}
    
    \clip[decorate,decoration={zigzag,amplitude=2,segment length=6.17}]
      (S) -- (-S) --++ (-1.1,-0.1) |-++ (4.2,2.2) |- cycle;
    \clip[decorate,decoration={zigzag,amplitude=2,segment length=6.17}]
      (-N) -- (N) --++ (1.1,0.1) |-++ (-4.2,-2.2) |- cycle;
    
    \fill[mylightpurple] (-N) |-++ (2,0.1) -- (N) -- (-S) -- (S) -- cycle;
    \fill[mylightpurple] (-S) |-++ (2,-0.1) -- (S) -- (-N) -- (N) -- cycle;
    
    \fill[mylightblue] (-N) -- (-E) -- (-S) -- (-W) -- cycle;
    \fill[mylightblue] (N) -- (E) -- (S) -- (W) -- cycle;
    
    
    \foreach \i [evaluate={\c=\i/(\Nlines+1); \cs=sin(90*\c);}] in {1,...,\Nlines}{
      \draw[world line,samples=\Nsamples,smooth,variable=\x,domain=0:2] 
        plot(\x-1,{1-2.2*kruskal(\x*pi/4,\cs-0.75)});
      
    }
    
    
  \end{scope}
  
  \draw[singularity] (-N) -- node[above] {future singularity $r=0$} (N);
  \draw[singularity] (S) -- node[below] {past singularity $r=0$} (-S);
  \path (S) -- (W) node[mydarkblue,pos=0.50,below=-2.5,rotate=-45,scale=0.85]
    {\contour{mylightpurple}{anti-horizon $r=2M$}};
  \path (W) -- (N) node[mydarkblue,pos=0.32,above=-2.5,rotate=45,scale=0.85]
    {\contour{mylightpurple}{horizon $r=2M$}};
  \draw[thick,mydarkblue] (-N) -- (-E) -- (-S) -- (-W) -- cycle;
  \draw[thick,mydarkblue] (N) -- (E) -- (S) -- (W) -- cycle;
  
  \node[fill=mylightblue,inner sep=2] at (-O) {III};
  \node[fill=mylightblue,inner sep=2] at (O) {I};
  \node[fill=mylightpurple,inner sep=2] at (0,0.64) {II};
  \node[fill=mylightpurple,inner sep=2] at (0,-0.64) {IV};
  
  \node[above=1,left=1,mydarkblue] at (-2.15,0) {$i^0$};
  \node[above=1,right=1,mydarkblue] at (2.15,0) {$i^0$};
  \node[right=1,below=1,mydarkpurple] at (-S) {$i^-$};
  \node[right=1,above=1,mydarkpurple] at (-N) {$i^+$};
  \node[right=1,below=1,mydarkpurple] at (S) {$i^-$};
  \node[right=1,above=1,mydarkpurple] at (N) {$i^+$};
  \node[mydarkblue,below left=-1] at (-1.5,-0.5) {$\calI^-$};
  \node[mydarkblue,above left=-1] at (-1.5,0.5) {$\calI^+$};
  \node[mydarkblue,above right=-1] at (1.5,0.5) {$\calI^+$};
  \node[mydarkblue,below right=-1] at (1.5,-0.5) {$\calI^-$};

  \node[fill=mylightpurple,inner sep=2] at (0.2,0.88) {$\psi(U)$};
  \node[fill=mylightpurple,inner sep=2] at (0.3,-0.6) {$\psi(p)$};
  
  
\end{tikzpicture}